\documentclass[10pt,a4paper,envcountsame]{llncs}
\usepackage{hyperref}
\usepackage{comment}
\usepackage{amsmath}
\usepackage{amsfonts}
\usepackage{amssymb}
\usepackage[capitalize]{cleveref}
\usepackage{thmtools,thm-restate}
\usepackage[textsize=small]{todonotes}
\usepackage{paralist}
\usepackage[ruled]{algorithm2e}
\usepackage{tikz}
\usepackage{subcaption}
\usepackage{booktabs}
\usepackage[all]{xy}

\usetikzlibrary{automata}
\usetikzlibrary{positioning}

\newcommand{\ignore}[1]{}

\newcommand{\cG}{\mathcal{G}}
\newcommand{\cA}{\mathcal{A}}
\newcommand{\cT}{\mathcal{T}}
\newcommand{\cM}{\mathcal{M}}
\newcommand{\cK}{\mathcal{K}}

\newcommand{\N}{\mathbb{N}}
\newcommand{\Z}{\mathbb{Z}}

\newcommand{\Q}{\mathbb{Q}}
\newcommand{\cV}{\mathcal{V}}
\newcommand{\cZ}{\mathcal{Z}}
\newcommand{\cC}{\mathcal{C}}

\newcommand{\Ncal}{\mathcal{N}}
\newcommand{\Nom}{\mathbb{N}_\omega}

\newcommand{\Lang}[1]{\mathsf{L}(#1)}
\newcommand{\Perms}{\Pi}
\newcommand{\Parikh}{\Psi}

\newcommand{\fin}{\textup{fin}}

\newcommand{\trans}[1]{\stackrel{#1}{\longrightarrow}}

\newcommand{\hurdle}{\textup{hurdle}}
\newcommand{\fdrop}{\textup{drop}}

\newcommand{\init}{\textup{init}}

\newcommand{\Loop}{\textup{Loop}}

\newcommand{\eff}{\textup{eff}}

\renewcommand{\vec}[1]{\mathbf{#1}}

\renewcommand{\enspace}{}

\newcommand{\sep}[2]{#1\mathrel{|}#2}

\newcommand{\scal}[2]{\langle #1, #2 \rangle}

\newcommand{\Pre}{\textsc{Pre}}
\newcommand{\Post}{\textsc{Post}}

\newcommand{\rev}[1]{{#1}^{\mathsf{rev}}}

\DeclareMathOperator{\cone}{cone}

\begin{document}


\title{An Approach to Regular Separability in Vector Addition Systems}

\author{Wojciech Czerwi\'nski\inst{1}
\and Georg Zetzsche\inst{2}}

\institute{University of Warsaw \and Max Planck Institute for Software Systems (MPI-SWS)}

%
{} 



\maketitle

\begin{abstract}
  We study the problem of regular separability of languages of vector
  addition systems with states (VASS).  It asks whether for two given
  VASS languages $K$ and $L$, there exists a regular language $R$ that
  includes $K$ and is disjoint from $L$. While decidability of the
  problem in full generality remains an open question, there are
  several subclasses for which decidability has been shown: It is
  decidable for (i)~one-dimensional VASS, (ii)~VASS coverability
  languages, (iii)~languages of integer VASS, and (iv)~commutative
  VASS languages.

  We propose a general approach to deciding regular separability. We
  use it to decide regular separability of an arbitrary VASS language
  from any language in the classes~(i), (ii), and~(iii). This
  generalizes all previous results, including~(iv).
\end{abstract}

\section{Introduction}\label{intro}

\paragraph{Vector addition systems with states} Vector addition
systems with states (VASS)~\cite{DBLP:journals/tcs/HopcroftP79} are
one of the most intensively studied model for concurrent systems.
They can be seen as automata with finitely many counters, which can be
increased or decreased whenever its values is non-negative, but not
tested for zero.  Despite their fundamental nature and the extensive
interest, core aspects remain obscure. A prominent example is the
reachability problem, which was shown decidable in the early
1980s~\cite{DBLP:conf/stoc/Mayr81}. However, its complexity remains
unsettled. The best known upper bounds are
non-primitive-recursive~\cite{DBLP:conf/lics/LerouxS15}, whereas the
best known lower bound is tower
hardness~\cite{DBLP:journals/corr/abs-1809-07115}, and reachability
seems far from being understood.

There is also a number of other natural problems concerning VASS where
the complexity or even decidability remains unresolved. An example is
the structural liveness problem, which asks whether there exists a
configuration such that for every configuration $c$ reachable from it
and every transition $t$ one can reach from $c$ some configuration in which $t$ is enabled.
Its decidability status was
settled only recently~\cite{DBLP:conf/sofsem/Jancar17}, but the complexity
is still unknown.  For closely related extensions of VASS, namely
branching VASS and pushdown VASS even decidability status is unknown
with the best lower bound being tower-hardness~\cite{LazicS15,LazicT17}.  This all suggests that there is still a lot to understand
about VASS.

\paragraph{Separability problem} One way to gain a fresh
perspective and deeper understanding of the matter is to study
decision problems that generalize reachability.  It seems to us that
here, a natural choice is the problem of \emph{regular separability}.
It asks whether for two given languages $K$ and $L$ there exists a
\emph{regular separator}, i.e. a regular language $R$ such that
$K\subseteq R$ and $R\cap L=\emptyset$.
Decidability of this problem for general VASS languages appears to be
difficult.  It has been shown decidable for several subclasses, namely
for (a)~\emph{commutative VASS
  languages}~\cite{DBLP:conf/stacs/ClementeCLP17} (equivalently,
separability of sections of reachability sets by recognizable sets),
for (b)~\emph{one-counter nets}~\cite{DBLP:conf/lics/CzerwinskiL17}
i.e.\ VASS with one counter, (c)~integer
VASS~\cite{DBLP:conf/icalp/ClementeCLP17}, i.e.\ VASS where we allow
counters to become negative, and finally for (d)~\emph{coverability
  languages}, which follows from the general decidability for
well-structured transition
systems~\cite{DBLP:conf/concur/CzerwinskiLMMKS18}. However, in full
generality, decidability remains a challenging open question.  It
should be mentioned that this line of research has already led to
unforeseen insights: The closely related problem of separability by
bounded regular languages prompted methods that turned out to yield
decidability results that were deeply unexpected~\cite{DBLP:conf/icalp/CzerwinskiHZ18}.

\paragraph{Contribution} We present a general approach to deciding
separability by regular languages and prove three new results, which
generalize all four regular separability results shown until
now. Specifically, we show decidability of regular separability of
(i)~VASS languages from languages of one-counter nets, (ii)~VASS
languages from coverability VASS languages, and (iii)~VASS languages
from integer VASS languages. This clearly generalizes results~(b),
(c), and (d) above, and we will see that this also strengthens~(a).

\paragraph{Main ingredients}
The starting point of our approach is the observation that for many
language classes $\cK$, deciding regular separability of a language
$L$ from a given language $K$ in $\cK$ can be reduced to deciding
regular separability of $L$ from some fixed language $G$ in $\cK$.  In
all three cases (i)--(iii), this allows us to interpret the
words in $L$ as walks in the grid $\Z^n$. For~(i), we then have to
decide separability from those walks in $\Z=\Z^1$ that remain in $\N$
and arrive at zero. For (ii), we decide separability from the set of
walks that remain in $\N^n$ and arrive somewhere in $\N^n$.  For
(iii), we want to separate from all walks in $\Z^n$ that end at the
origin. The corresponding fixed languages are denoted $D_1$ (for~(i)),
$C_n$ (for~(ii)), and $Z_n$ (for~(iii)), respectively.

%
%
%
In order to decide separability from a fixed language $G$ (i.e. $D_1$,
$C_n$, or $Z_n$), we first classify those regular languages that are
disjoint from $G$.  Second, the classifications are used to decide
whether a given VASS language $L$ is included in such a regular
language. These decision procedures employ either the previous
result~(a) above or reduce to the simultaneous unboundedness problem
(which is known to be decidable for VASS
languages~\cite{DBLP:conf/icalp/HabermehlMW10,DBLP:conf/icalp/CzerwinskiHZ18}).

\paragraph{VASS vs. integer VASS} The result~(iii) is significantly
more involved than~(i) and~(ii).  First, the classification of regular
languages disjoint from $Z_n$ leads to a geometric characterization of
regular separability. This is then applied in a decision procedure
that employs the KLMST decomposition from the algorithms by Sacerdote
and Tenney~\cite{sacerdote1977decidability},
Mayr~\cite{DBLP:conf/stoc/Mayr81},
Kosaraju~\cite{DBLP:conf/stoc/Kosaraju82}, and
Lambert~\cite{DBLP:journals/tcs/Lambert92} (and recast by Leroux and
Schmitz~\cite{DBLP:conf/lics/LerouxS15}) for reachability in VASS.
Previous algorithms for VASS languages that use this decomposition (by
Habermehl, Meyer, and Wimmel~\cite{DBLP:conf/icalp/HabermehlMW10} and
by Czerwi\'{n}ski, Hofman, and
Zetzsche~\cite{DBLP:conf/icalp/CzerwinskiHZ18}) perform the
decomposition once, which yields regular overapproximations that
contain all information needed for their purposes. Our procedure
requires an additional refinement: Depending on a property of each
overapproximation, we can either reduce separability to the
commutative case and apply~(a) or we can reduce the dimension of the
input language (i.e. transform it into a set of walks in $\Z^m$ for
$m<n$) and invoke our algorithm recursively.



\paragraph{Connection to VASS reachability}
We hope that this approach can be used to decide regular separability
for VASS in full generality in the future. This would amount to
deciding regular separability of a given VASS language from the set of
all walks in $\Z^n$ that remain in $\N^n$ and arrive in the
origin. The corresponding language is denoted $D_n$. We emphasize that
an algorithm along these lines might directly yield new insights
concerning reachability: Classifying those regular
languages that are disjoint from $D_n$ would yield an algorithm for
reachability because the latter reduces to intersection of a given
regular language with $D_n$. Such an algorithm would look for a
certificate for \emph{non-reachability} (like Leroux's
algorithm~\cite{DBLP:conf/lics/Leroux09}) instead of a run.


\paragraph{Related work} Aside from \emph{regular} separability,
separability problems in a more general sense have also attracted
significant attention in recent years.  Here, the class of sought
separators can differ from the regular languages.  A series of recent
works has concentrated on separability of regular languages by
separators from
subclasses~\cite{DBLP:conf/lics/Place15,DBLP:conf/mfcs/PlaceRZ13,DBLP:conf/stacs/PlaceZ15,DBLP:journals/corr/PlaceZ14,DBLP:conf/csr/PlaceZ17,DBLP:conf/lics/PlaceZ17},
and work in this direction has been started for trees as
well~\cite{DBLP:journals/fuin/Bojanczyk17,DBLP:conf/icalp/Goubault-Larrecq16}.


In the case of non-regular languages as input languages, it was shown
early that regular separability  is
undecidable for context-free languages~\cite{SzymanskiW-sicomp76,DBLP:journals/jacm/Hunt82a}.
Moreover, aside from the above mentioned results on regular separability,
infinite-state systems have also been studied with respect to separability
by bounded regular languages~\cite{DBLP:conf/icalp/CzerwinskiHZ18} and
piecewise testable languages~\cite{DBLP:journals/dmtcs/CzerwinskiMRZZ17}
and generalizations thereof~\cite{DBLP:conf/lics/Zetzsche18}.




\section{Preliminaries}\label{prelims}
By $\Q$ ($\Q_+$), we denote the set of (non-negative) rational
numbers. Let $\Sigma$ be an alphabet and let $\varepsilon$ denote the empty
word. If $\Sigma=\{x_1,\ldots,x_n\}$, then the \emph{Parikh image} of
a word $w\in\Sigma^*$ is defined as
$\Parikh(w)=(|w|_{x_1},\ldots,|w|_{x_n})$, where $|w|_x$ denotes the
number of occurrences of $x$ in $w$. The \emph{commutative closure} of a language $L\subseteq\Sigma^*$ is $\Perms(L)=\{u\in\Sigma^* \mid \exists v\in L\colon
\Parikh(v)=\Parikh(u)\}$. 

A \emph{($n$-dimensional) vector addition system with states
  (VASS)} is a tuple
$V=(Q,T,s,t)$, where $Q$ is a finite set of \emph{states},
$T\subseteq Q\times \Sigma_\varepsilon\times \Z^n\times Q$ is a finite
set of transitions, $s\in Q$ is its \emph{source state}, $t\in Q$ is
its \emph{target state}. Here, $\Sigma_\varepsilon$ denotes
$\Sigma\cup\{\varepsilon\}$. A \emph{configuration} of $V$ is a pair
$(q,\vec u)\in Q\times\N^n$. For each transition
$t = (q,x,\vec v,q')\in T$ and configurations $(q,\vec u)$,
$(q',\vec u')$ with $\vec u'=\vec u+\vec v$, we write
$(q,\vec u)\xrightarrow{x} (q',\vec u')$.  For a word $w\in\Sigma^*$,
we write $(q,\vec u)\trans{w}(q',\vec u')$ if there are
$x_1,\ldots,x_n\in\Sigma_\varepsilon$ and configurations
$(q_i,\vec v_i)$ for $i\in[0,n]$ with
$(q_{i-1},\vec v_{i-1})\xrightarrow{x_i}(q_i,\vec v_i)$ for
$i\in[1,n]$, $(q_0,\vec v_0)=(q,\vec u)$, and
$(q_n,\vec v_n)=(q',\vec u')$. The language of $V$ is then
$\Lang{V}=\{w\in\Sigma^* \mid (s,\textbf{0})\trans{w}(t,\textbf{0})\}$. An
\emph{($n$-dimensional) integer vector addition system with states
  ($\Z$-VASS)}~\cite{DBLP:conf/rp/HaaseH14} is syntactically a VASS,
but for $\Z$-VASS, the configurations are pairs in $Q\times\Z^n$.
This difference aside, the language is defined verbatim.  Likewise, an
\emph{$n$-dimensional coverability vector addition system with states
  (coverability VASS)} is syntactically a VASS. However, if we regard
a VASS $V$ as a coverability VASS, we define its language as
$\Lang{V}=\{w\in\Sigma^*\mid \text{$(s,\textbf{0})\trans{w}(t,\vec u)$ for some
  $\vec u\in\N^n$}\}$.  Let $\cV_n$ ($\cZ_n$, $\cC_n$) denote the
class of languages of $n$-dim.\ VASS ($\Z$-VASS, coverability VASS).

Let $\Sigma_n=\{a_i,\bar{a}_i\mid i\in[1,n]\}$ and define the
homomorphism $\varphi_n\colon\Sigma_n^*\to\Z^n$ by
$\varphi_n(a_i)=\vec e_i$ and $\varphi_n(\bar{a}_i)=-\vec e_i$. Here,
$\vec e_i\in\Z^n$ is the vector with $1$ in coordinate $i$ and $0$
everywhere else.  By way of $\varphi_n$, we can regard words from
$\Sigma_n^*$ as walks in the grid $\Z^n$ that start in the origin.
Later, we will only write $\varphi$ when the $n$ is clear from the
context.  With this, let $Z_n=\{w\in\Sigma^*_n\mid
\varphi(w)=\textbf{0}\}$. Hence, $Z_n$ is the set of walks that start and end
in the origin.

For $w\in\Sigma_1^*$, let
$\fdrop(w)=\min\{\varphi(v)\mid \text{$v$ is a prefix of
  $w$}\}$. Thus, if $w$ is interpreted as walking along $\Z$, then
$\fdrop(w)$ is the lowest value attained on the way. Note that
$\fdrop(w)\in [-|w|,0]$ for every $w\in\Sigma_1^*$. We define
$C_1=\{w\in\Sigma_1^* \mid \fdrop(w)=0\}$. For each $i\in[1,n]$, let
$\lambda_i\colon\Sigma_n^*\to\Sigma_1^*$ be the homomorphism with
$\lambda_i(a_i)=a_1$, $\lambda_i(a_j)=\varepsilon$ for $j\ne i$, and
$\lambda_i(\bar{a}_j)=\overline{\lambda_i(a_j)}$ for every
$j\in[1,n]$. Then we define $C_n=\bigcap_{i=1}^n
\lambda_i^{-1}(C_1)$. Thus, $C_n$ is the set of walks in $\Z^n$ that
start in the origin and remain in the positive orthant
$\N^n$. Finally, let $D_n=C_n\cap Z_n$. Hence, $D_n$ collects those
walks that start in the origin, always remain in $\N^n$ and arrive in
the origin.  For $w\in\Sigma_n^*$, $w=w_1\cdots w_m$,
$w_1,\ldots,w_m\in\Sigma_n$, let $\bar{w}=\bar{w}_1\cdots\bar{w}_m$ and
$\rev{w}=w_m\cdots w_1$.  Here, we set $\bar{\bar{a}}_i=a_i$ for
$a_i\in\Sigma_n$. For $L \subseteq \Sigma^*$ we define $\overline{L} = \{\bar{w} \mid w \in L\}$.

For alphabets $\Sigma,\Gamma$, a subset
$T\subseteq\Sigma^*\times\Gamma^*$ is a \emph{rational transduction}
if it is a homomorphic image of a regular language, i.e.\ if there is
an alphabet $\Delta$, a regular $K\subseteq\Delta^*$, and a
morphism $h\colon\Delta^*\to\Sigma^*\times\Gamma^*$ such that
$T=h(K)$. Typical examples of rational transductions are the relation
$\{(w,g(w)) \mid w\in\Sigma^*\}$ for some morphism
$g\colon \Sigma^*\to\Gamma^*$ or $\{(w,w) \mid w\in R\}$ for some
regular language $R\subseteq\Sigma^*$~\cite{Ber79}.

It is well-known that if $S\subseteq\Sigma^*\times\Gamma^*$ and
$T\subseteq\Delta^*\times\Sigma^*$ are rational transductions, then
the relation $S\circ T$, which is defined
$\{(u,v)\in\Delta^*\times\Gamma^*\mid \exists w\in\Sigma^*\colon
(u,w)\in T,~(w,v)\in S\}$ and also
$T^{-1}=\{(v,u)\in\Sigma^*\times\Delta^* \mid (u,v)\in T\}$ are
rational transductions as well~\cite{Ber79}.

For a language $L\subseteq \Sigma^*$ and a subset
$T\subseteq\Sigma^*\times\Gamma^*$, we define
$TL=\{v\in\Gamma^*\mid \exists u\in L\colon (u,v)\in T\}$.  
A language class $\cK$ is called \emph{full trio} if for every
$L\subseteq\Sigma^*$ from $\cK$, and every rational transduction
$T\subseteq\Sigma^*\times\Gamma^*$, we also have $TL$ in $\cK$.  The
full trio \emph{generated by $L$}, denoted by $\cM(L)$, is the class
of all languages $TL$, where $T\subseteq\Sigma^*\times\Gamma^*$ is a
rational transduction for some $\Gamma$. It is well-known that
$\cV_n$, $\cC_n$, and $\cZ_n$ are (effectively) the full trios generated
by $D_n$, $C_n$, and $Z_n$,
respectively~\cite{Greibach1978,Jantzen1979}.
Since these two paper do not mention effectivity explicitly, we include a short proof.

\begin{proposition}\label{full-trio-generators}
  We have the identities $\cV_n=\cM(D_n)$, $\cC_n=\cM(C_n)$, and
  $\cZ_n=\cM(Z_n)$. Moreover, all inclusions are effective: A
  description in one form can be effectively transformed into the
  other.
\end{proposition}
\begin{proof}
  To simplify notation, we use finite-state transducers in the proof.
  A \emph{finite-state transducer} is a tuple
  $\cT=(Q,\Sigma,\Gamma,E,s,t)$, where $Q$ is a finite set of
  \emph{states}, $\Sigma$ is its \emph{input alphabet}, $\Gamma$ is
  its \emph{output alphabet},
  $E\subseteq Q\times \Sigma^*\times\Gamma^*\times Q$ is a finite set
  of \emph{transitions}, $s\in Q$ is its \emph{starting state}, and
  $t\in Q$ is its \emph{terminal state}. For such a transducer and a
  language $L\subseteq\Sigma^*$, we write $R(\cT)$ for the set of all
  pairs $(u,v)\in\Sigma^*\times\Gamma^*$
  for which there is a sequence of transitions
  $(q_0,u_1,v_1,q_1)(q_1,u_2,v_2,q_2)\cdots (q_{n-1},u_n,v_n,q_n)$ such that
  $q_0=s$, $q_n=t$, $u=u_1\cdots u_n$, and $v=v_1\cdots v_n$.  It is
  easy to see that a relation $T\subseteq\Sigma^*\times\Gamma^*$ is
  rational by our definition if and only if there is a finite-state
  transducer $\cT$ with $T=R(\cT)$ \cite[Theorem~6.1]{Ber79}.  If
  $T=R(\cT)$ and $L\subseteq\Sigma^*$, then for $TL$, we also simply
  write $\cT(L)$.
  
  We begin with the inclusion $\cM(D_n)\subseteq\cV_n$.  Given a
  transducer $\cT=(Q,\Sigma_n,\Gamma,E,s,t)$ and $n\in\N$, we
  construct a VASS $V$ as follows. First, by splitting each edges into
  a sequence of edges, we may assume that every edge in $T$ is of the
  form $(p,w,x,q)$ with $x\in\Gamma\cup\{\varepsilon\}$.  We define
  $V=(Q,\Gamma,T,s,t)$ as follows. For every edge $(p,w,x,q)$ in
  $\cT$, $V$ has a transition $(p,x,\vec w,q)$, where
  $\vec w=\varphi(w)$.  Recall that $\varphi\colon\Sigma_n^*\to\Z^n$
  is the morphism with $\varphi(a_i)=\vec e_i$ and
  $\varphi(\bar{a}_i)=-\vec e_i$, where $\vec e_i\in\Z^n$ is the unit
  vector with $1$ in the $i$-th coordinate and $0$ everywhere else.
  Then clearly, $\Lang{V}=\cT(D_n)$. This proves
  $\cM(D_n)\subseteq\cV_n$.  For the inclusions
  $\cM(C_n)\subseteq\cC_n$ and $\cM(Z_n)\subseteq\cZ_n$, we construct
  the same VASS, but interpret it as a coverability VASS or as an
  integer VASS and the same equality of languages will hold.
  
  For the inclusion $\cV_n\subseteq\cT(D_n)$, consider a VASS
  $V=(Q,\Gamma,T,s,t)$. To construct the transducer
  $\cT=(Q,\Sigma_n,\Gamma,E,s,t)$ so that $\cT(D_n)=\Lang{V}$, we
  need to turn vectors $\vec u\in\Z^n$ into words. Given a vector
  $\vec u=(u_1,\ldots,u_n)\in\Z$, define the word $w_{\vec u}$ as
  $a_1^{u_1}\cdots a_n^{u_n}$. Here, if $u_i<0$, we define $a_i^{u_i}$
  as $\bar{a}_i^{|u_i|}$. Our transducer has the following edges. For
  each transition $(p,x,\vec u,q)$ in $V$, $\cT$ has an edge
  $(p, w_{\vec u}, x,q)$.  Then clearly, $\cT(D_n)=\Lang{V}$.  As
  above, if $\cT$ is applied to $C_n$ or $Z_n$, then we construct the
  same VASS, but interpret it as a coverability VASS or integer VASS,
  respectively.
  \qed
\end{proof}

\paragraph{State of the art}
We now give a brief overview of previous results on regular
separability for subclasses of VASS languages. Two languages
$K,L\subseteq\Sigma^*$ are called \emph{regularly separable} if there
exists a regular language $S\subseteq\Sigma^*$ with $K\subseteq S$ and
$L\cap S=\emptyset$. In that case, we write $\sep{K}{L}$.  The
\emph{regular separability problem} asks, given languages $K$ and $L$,
whether $\sep{K}{L}$.  The first studied subclass of VASS was that of \emph{commutative VASS languages}, i.e. those of the form $\Perms(L)$ for a VASS language $L$.
\begin{theorem}[\cite{DBLP:conf/stacs/ClementeCLP17}]\label{separability-commutative}
  Given VASS languages $K,L\subseteq\Sigma^*$, it is decidable whether
  $\sep{\Perms(K)}{\Perms(L)}$.
\end{theorem}

As observed in \cite{DBLP:conf/icalp/CzerwinskiHZ18},
\cref{separability-commutative} also implies the following.
\begin{corollary}[\cite{DBLP:conf/stacs/ClementeCLP17,DBLP:conf/icalp/CzerwinskiHZ18}]\label{separability-bounded}
  Given VASS languages $K,L$ and words $w_1,\ldots,w_m\in\Sigma^*$
  such that $K,L\subseteq w_1^*\cdots w_m^*$, it is decidable whether
  $\sep{K}{L}$.
\end{corollary}
\noindent
After \cref{separability-commutative}, the investigation went on to 
1-dim.\ VASS~\cite{DBLP:conf/lics/CzerwinskiL17}:
\begin{theorem}[\cite{DBLP:conf/lics/CzerwinskiL17}]\label{separability-1vass}
  Given 1-VASS languages $K$ and $L$, it is decidable whether
  $\sep{K}{L}$.
\end{theorem}
\noindent
Moreover, the next \lcnamecref{separability-zvass} has been
established in \cite{DBLP:conf/icalp/ClementeCLP17}.
\begin{theorem}[\cite{DBLP:conf/icalp/ClementeCLP17}]\label{separability-zvass}
  Given $\Z$-VASS languages $K,L\subseteq\Sigma^*$, it is decidable
  whether $\sep{K}{L}$.
\end{theorem}
It should be noted that the authors of~\cite{DBLP:conf/icalp/ClementeCLP17} speak of Parikh automata, but
these are equivalent to $\Z$-VASS: Parikh automata are equivalent to
reversal-bounded counter
machines~\cite[Prop.~3.13]{DBLP:journals/ita/CadilhacFM12} and
the latter are equivalent to blind counter
machines~\cite[Theorem~2]{Greibach1978}, which are the same as $\Z$-VASS.

Finally, a recent general result shows that any two coverability
languages of well-structured transition
systems~\cite{FinkelSchnoebelen2001,geeraerts2007well} fulfilling some mild conditions are regular
separable if and only if they are
disjoint~\cite{DBLP:conf/concur/CzerwinskiLMMKS18}.
In particular it applies to the situation when the systems are upward-compatible and one of them
is finitely branching, which is the case for coverability VASS languages:
\begin{theorem}[\cite{DBLP:conf/concur/CzerwinskiLMMKS18}]\label{separability-cover}
  Given coverability VASS languages $K,L\subseteq\Sigma^*$, it is decidable
  whether $\sep{K}{L}$.
\end{theorem}

\section{Main Results}\label{results}

\newcommand{\TOneVASS}{$1$-VASS}
\newcommand{\TCoverabilityVASS}{$\ge$-VASS}
\newcommand{\TZVASS}{$\Z$-VASS}
\newcommand{\TCommutativeVASS}{c-VASS}
\newcommand{\TVASS}{VASS}
\newcommand{\TDecidableKnown}{D}
\newcommand{\TDecidableNew}{\textbf{D}}

\begin{table}
  \begin{tabular}{llllll}
    \toprule
    \TOneVASS & \TCoverabilityVASS & \TZVASS & \TCommutativeVASS & \TVASS & \\
    \midrule
    \TDecidableKnown~\cite{DBLP:conf/lics/CzerwinskiL17} & \TDecidableNew         & \TDecidableNew          & \TDecidableNew          & \TDecidableNew     & \TOneVASS \\
              & \TDecidableKnown~\cite{DBLP:conf/concur/CzerwinskiLMMKS18}         & \TDecidableNew          & \TDecidableNew          & \TDecidableNew     & \TCoverabilityVASS \\
           &    & \TDecidableKnown~\cite{DBLP:conf/icalp/ClementeCLP17}          & \TDecidableNew          & \TDecidableNew     & \TZVASS \\
    
             &          &        &\TDecidableKnown~\cite{DBLP:conf/stacs/ClementeCLP17}          & ?           & \TCommutativeVASS \\
              &          &           &           & ?     & \TVASS \\
    \bottomrule
  \end{tabular}
  \caption{Overview of the decidability of regular separability for
    VASS subclasses. Here, \TCoverabilityVASS\ and \TCommutativeVASS\
    are short for coverability VASS languages, and commutative VASS
    languages, respectively. The entry in the column for class
    $\mathcal{K}_0$ and the row for class $\mathcal{K}_1$ denotes
    decidability of regular separability of languages of
    $\mathcal{K}_0$ from languages of $\mathcal{K}_1$. The entries in
    bold are new consequences of results in this
    paper.}\label{table-results}
\end{table}

In this section, we record the main results of this work. See
\cref{table-results} for an overview. Our first main result is that
regular separability is decidable if one input language is a VASS
language and the other is the language of a 1-VASS.
\begin{theorem}\label{result-vass-onevass}
  Given a VASS $V_0$ and a 1-dim.\  VASS $V_1$, it is decidable
  whether $\sep{\Lang{V_0}}{\Lang{V_1}}$.
\end{theorem}
This generalizes \cref{separability-1vass}, because here, one of the
input languages can be an arbitrary VASS language.
Our second main result generalizes \cref{separability-cover} in the
same way as \cref{result-vass-onevass} extends
\cref{separability-1vass}:
\begin{theorem}\label{result-vass-cover}
  Given a VASS $V_0$ and a coverability VASS $V_1$, it is decidable
  whether $\sep{\Lang{V_0}}{\Lang{V_1}}$.
\end{theorem}
Our third main result is decidability of regular separability of a
given VASS language from a given $\Z$-VASS language. 
\begin{theorem}\label{result-vass-zvass}
  Given a VASS $V_0$ and a $\Z$-VASS $V_1$, it is decidable whether
  $\sep{\Lang{V_0}}{\Lang{V_1}}$.
\end{theorem}
As before, this significantly generalizes \cref{separability-zvass}.
Our proof of \cref{result-vass-zvass} relies on
\cref{separability-commutative}. At first glance, it might seem that
\cref{result-vass-zvass} is unrelated to regular separability of
commutative VASS languages. However, a simple observation shows that
\cref{result-vass-zvass} also strengthens
\cref{separability-commutative}. This is because for deciding regular
separability of commutative VASS languages, one may assume that one of
the input languages is $Z_n$:
\begin{proposition}\label{commutative-one-side-zn}
  Let $\Gamma_n=\{a_1,\ldots,a_n\}\subseteq\Sigma_n$.
  For any $K,L\subseteq\Gamma_n^*$, we have
  $\sep{\Perms(K)}{\Perms(L)}$ if and only if
  $\sep{\Perms(K\overline{L})}{Z_n}$.
\end{proposition}

Since $\Perms(K\overline{L})$ is a VASS language and $Z_n$ is a
$\Z$-VASS language, this means \cref{result-vass-zvass} indeed
strengthens \cref{separability-commutative}.

The rest of this section proves \Cref{commutative-one-side-zn}, which
follows from two simple observations.  The first concerns separability
of subsets of monoids.  If $M$ is a monoid and $K,L\subseteq M$ are
subsets, then $K$ and $L$ are called \emph{separable} if there is a
morphism $\varphi\colon M\to F$ into a finite monoid $F$ such that
$\varphi(K)\cap\varphi(L)=\emptyset$. Clearly, if
$K,L\subseteq\Sigma^*$, then $K$ and $L$ are separable if and only if
$\sep{K}{L}$. Therefore, it creates no inconsistencies to write
$\sep{K}{L}$ whenever $K,L\subseteq M$ are separable.  Let
$\Delta=\{(m,m)\mid m\in M\}\subseteq M\times M$.
\begin{lemma}\label{delta-separability}
  Let $K,L\subseteq M$. Then $\sep{K}{L}$ if and only if
  $\sep{K\times L}{\Delta}$.
\end{lemma}
\begin{proof}
  If $\sep{K}{L}$ with $\varphi\colon M\to F$, 
  define $\varphi'\colon M\times M\to F\times F$ by
  $\varphi'(u,v)=(\varphi(u),\varphi(v))$. Then clearly
  $\varphi'(K\times L)\cap \varphi'(\Delta)=\emptyset$. Conversely, if
  $\sep{K\times L}{\Delta}$ with a morphism
  $\varphi\colon M\times M\to F$, let $\varphi'\colon M\to F$ be
  the morphism with $\varphi'(u)=\varphi(u,1)$ for $u\in M$. Then we have
  $\varphi'(K)\cap\varphi'(L)=\emptyset$, because if there were $u\in K$, $v\in L$ with $\varphi'(u)=\varphi'(v)$, then
  \begin{multline*}
    \varphi(K\times
    L)\ni\varphi(u,v)=\varphi(u,1)\varphi(1,v) \\
    =\varphi(v,1)\varphi(1,v)=\varphi(v,v)\in\varphi(\Delta),
  \end{multline*}
  and thus $\varphi(K\times L)\cap\varphi(\Delta)\ne\emptyset$, which
  is impossible.
  \qed
\end{proof}

For subsets $S,T\subseteq\N^\Sigma$, separability is equivalent to
unary separability as studied by
\cite{DBLP:conf/stacs/ClementeCLP17}. We now have:
\begin{align*}
  \sep{\Perms(K)}{\Perms(L)} &\Leftrightarrow \sep{\Parikh(K)}{\Parikh(L)} \Leftrightarrow \sep{\Parikh(K)\times\Parikh(L)}{\Delta} \\
                               & \Leftrightarrow \sep{\Parikh(K)\times\Parikh(\bar{L})}{\Delta'}\Leftrightarrow\sep{\Perms(K\bar{L})}{Z_n},
\end{align*}
where $\Delta'=\{\vec u\in\N^{\Sigma\cup\bar{\Sigma}}\mid \text{$\vec u(a_i)=\vec u(\bar{a}_i)$ for $i\in[1,n]$}\}$.
Here, the third equivalence is just renaming components.
The second equivalence is \cref{delta-separability}; the first
and last equivalence are due to an observation
from~\cite{DBLP:conf/stacs/ClementeCLP17}: In \cite[Lemma
11]{DBLP:conf/stacs/ClementeCLP17}, it is shown that for languages
$K,L\subseteq\Sigma^*$, we have $\sep{\Perms(K)}{\Perms(L)}$ if and
only if $\sep{\Parikh(K)}{\Parikh(L)}$.  This completes 
\cref{commutative-one-side-zn}.

\section{VASS vs. 1-VASS}\label{vass-1vass}

In this section, we introduce our approach to regular separability
together with the first application: Regular separability of VASS
languages and 1-dim.\ VASS languages.

Our approach is inspired by the decision procedure for regular
separability for one dimensional VASS~\cite{DBLP:conf/lics/CzerwinskiL17}.
There, given languages $K$ and $L$, the idea is to construct
\emph{approximants} $K_k$ and $L_k$ for $k\in\N$. Here, $K_k$ and
$L_k$ are regular languages with $K\subseteq K_k$ and $L\subseteq L_k$
for which one can show that $\sep{K}{L}$ if and only if there is a
$k\in\N$ with $K_k\cap L_k=\emptyset$. The latter condition is then
checked algorithmically.

We simplify this idea in two ways. First, we show that for many
language classes, one may assume that one of the two input languages
is fixed (or fixed up to a parameter). Roughly speaking, if a language
class $\cK$ is defined by machines involving a finite-state control, then
 $\cK$ is typically a full trio since a
transduction can be applied using a product construction in the
finite-state control. Moreover, there is often a simple set $\cG$ of
languages so that $\cK$ is the full trio generated by $\cG$. For example,
as mentioned above, $\cV_n$ is generated by $D_n$ for each $n\ge
1$. This makes the following simple \lcnamecref{movetrans} very
useful.
\begin{lemma}\label{movetrans}
Let $T$ be a rational transduction. Then $\sep{L}{TK}$ if and only if $\sep{T^{-1}L}{K}$.
\end{lemma}
\begin{proof}
Suppose $L\subseteq R$ and $R\cap TK=\emptyset$ for some regular
$R$. Then clearly $T^{-1}L\subseteq T^{-1}R$ and $T^{-1}R\cap
K=\emptyset$. Therefore, the regular set $T^{-1}R$ witnesses
$\sep{T^{-1}L}{K}$. Conversely, if $\sep{T^{-1}L}{K}$, then
$\sep{K}{T^{-1}L}$ and hence, by the first direction,
$\sep{(T^{-1})^{-1}K}{L}$. Since $(T^{-1})^{-1}=T$, this reads
$\sep{TK}{L}$ and thus $\sep{L}{TK}$.
\qed
\end{proof}
Suppose we have full trios $\cK_0$ and $\cK_1$ generated by languages
$G_0$ and $G_1$, respectively.  Then, to decide if
$\sep{T_0G_0}{T_1G_1}$, we can check whether
$\sep{T_1^{-1}T_0G_0}{G_1}$.  Since $T_1^{-1}T_0$ is also a rational
transduction and hence $T_1^{-1}T_0G_0$ belongs to $\cK_0$, this means
we may assume that one of the input languages is $G_1$. This effectively
turns separability into a decision problem with one input language $L$
where we ask whether $\sep{L}{G_1}$.

Going further in this direction, instead of considering approximants
of two languages, we just consider regular overapproximations of $G_1$
and decide whether $L$ intersects all of them. However, we find it
more convenient to switch to the complement and think in terms of
``basic separators of $G_1$'' instead of overapproximations of
$G_1$. Informally, we call a family of regular languages \emph{basic
  separators of $G_1$} if (i)~each of them is disjoint from $G_1$ and
(ii)~every regular language $R$ disjoint from $G_1$ is included in a
finite union of basic separators.  This implies that $\sep{L}{G_1}$ if
and only if there exists a finitely many basic separators
$S_1, \ldots, S_k$ such that $L$ is contained in the union
$\bigcup_{i \in [1,k]} S_i$. Note that for each language $G_1$ there
trivially exists a family of basic separators; just take the family of
all regular languages disjoint from $G_1$. Our approach is to identify
a family of basic separators for which it is decidable whether a
language from $\cK_0$ is included in a finite union of them.

\paragraph{Basic separators for one-dimensional VASS} Let us see
this approach in an example and prove
\cref{result-vass-onevass}. Since $\cV_1$ is generated as a full trio
by $D_1$, \cref{movetrans} tells us that it suffices to decide whether
a given VASS language $L$ fulfills $\sep{L}{D_1}$. Now the first step is to
develop a notion of basic separators for $D_1$.


Since $D_1\subseteq\Sigma_1^*$, we assume now that $n=1$, meaning
$\varphi\colon\Sigma_1^*\to \Z$. One way a
finite automaton can guarantee non-membership in $D_1$ is by modulo
counting.  For $k\in\N$, let
\[ M_k=\{w\in\Sigma_1^* \mid \varphi(w)\not\equiv 0\bmod{k}\},\]
which is regular. Another option for an automaton to make
sure an input word $w$ avoids $D_1$ is to guarantee (i)~for
prefixes $v$ of $w$, that $\varphi(v)$ does not exceed some $k$ if
$\fdrop(v)=0$ and (ii)~$\varphi(w) \neq 0$.  For $w\in\Sigma_1^*$, let
$\mu(w)=\max\{\varphi(v) \mid \text{$v$ is a prefix of $w$ and
  $\fdrop(v)=0$}\}$ and
\[ B_{k}=\{w\in\Sigma_1^* \mid \text{$w\notin D_1$ and $\mu(w)\le k$}\} \]
Here, the B stands for ``bounded counter value''.
It is obvious that the languages $B_k$ are disjoint from $D_1$. We observe that they are regular:
A word $w$ with $\mu(w)\le k$ avoids $D_1$ if and only if
(i)~$\varphi$ drops below zero after a prefix where $\varphi$ is
confined to $[0,k]$ or (ii)~$\varphi$ stays above zero and thus
assumes values in $[0,k]$ throughout. The third type of separator
is a symmetric right-to-left version of $B_{k}$, namely 
\begin{align*}
  \rev{\bar{B}_{k}}&=\{\rev{\bar{w}} \mid w\in B_{k}\} \\
                   &=\{w\in \Sigma_1^* \mid \text{$w\notin D_1$ and $\mu(\rev{\bar{w}})\le k$}\}
\end{align*}
Then we have indeed a family of basic separators for $D_1$:
%
%
%
%
\begin{restatable}{lemma}{basicSeparatorsOneVASS}\label{basic-separators-1vass}
  Let $R\subseteq\Sigma_1^*$ be a regular language. Then
  $R\cap D_1=\emptyset$ if and only if $R$ is included in
  $M_k\cup B_\ell\cup\rev{\bar{B}_{m}}$ for some $k,\ell,m\in\N$.
\end{restatable}


\begin{proof}
The proof of \cref{basic-separators-1vass} is very similar to the proof
of Lemma~8 in~\cite{DBLP:conf/lics/CzerwinskiL17}, but phrased in a slightly different setting.

  The ``if'' direction is obvious, so let us prove the ``only if''.
  Suppose $R\cap D_1=\emptyset$ and $R=\Lang{\cA}$ for an automaton
  $\cA$ with $n$ states. We claim that then
  $R\subseteq M_{n!}\cup B_{n}\cup\rev{\bar{B}_{n}}$.

  We proceed by a relatively simple pumping argument.
  Towards a contradiction, we assume that there is a word $w\in R$
  with $w\notin M_{n!}\cup B_{n}\cup \rev{\bar{B}_{n}}$.
  %
  %
  This means $\varphi(w)\equiv 0\bmod{n!}$ and $w$ has a prefix $u'$
  with $\fdrop(u')=0$ and $\varphi(u')=\mu(w)>n$ and a suffix $v'$
  with $\fdrop(\rev{\bar{v'}})=0$ and
  $\varphi(\rev{\bar{v'}})=\mu(\rev{\bar{w}})>n$, hence
  $\varphi(v')<-n$. We aim at pumping $u'$ and $v'$ such we get a word in $D_1$
  and finish with contradiction.
  Let $u$ be the shortest prefix of $w$ with
  $\varphi(u)=\varphi(u')$ and let $v$ be the shortest suffix with
  $\varphi(v')=\varphi(v)$. Then $|u|\le |u'|$ and $|v|\le|v'|$, which
  means in particular $\fdrop(u)=0$ and $\fdrop(\rev{\bar{v}})=0$.

  Let us show that $u$ and $v$ do not overlap in $w$, i.e.
  $|w|\ge |u|+|v|$. If they do overlap, we can write $w=xyz$ so that
  $u=xy$ and $v=yz$ with $y\ne\varepsilon$. Then by minimality of $u$,
  we have $\varphi(x)<\varphi(xy)$ and thus
  $\varphi(y)>0$. Symmetrically, minimality of $v$ yields
  $\varphi(\rev{\bar{z}})<\varphi(\rev{\bar{yz}})$ and thus
  $-\varphi(y)=\varphi(\rev{\bar{y}})>0$, contradicting
  $\varphi(y)>0$. Thus $u$ and $v$ do not overlap and we can write
  $w=uw'v$.
  
  Since $\varphi(u)>n$, we can decompose $u=u_1u_2u_3$ so that
  $1\le\varphi(u_2)\le n$ and in the run of $\cA$ for $w$, $u_2$ is read
  on a cycle. Analogously, since $\varphi(v)<-n$, we can decompose
  $v=v_1v_2v_3$ so that $-n\le\varphi(v_2)\le -1$ and $v_2$ is read on a
  cycle.
  
  Since $\varphi(w)\equiv 0\bmod{n!}$ and $\varphi(u_2)\in[1,n]$ and
  $\varphi(v_2)\in[-n,-1]$, there are $p,q\in\N$ with
  $\varphi(w)+p\varphi(u_2)+q\varphi(v_2)=0$. Moreover, we also
  have
  \begin{equation} \varphi(w)+(p+r|\varphi(v_2)|)\varphi(u_2)+(q+r|\varphi(u_2)|)\varphi(v_2)=0 \label{effect-zero} \end{equation}
  for every $r\in\N$. Consider the word
  \[ w_r = u_1u_2^{p+r|\varphi(v_2)|}u_3 w'v_1v_2^{q+r|\varphi(u_2)|}v_3. \]
  
  Since $u_2$ and $v_2$ are read on cycles, we have $w_r \in
  R$. Moreover, \cref{effect-zero} tells us that $\varphi(w_r)=0$.
  Finally, since $\fdrop(u)=0$ and $\varphi(u_2)>0$, for large
  enough $r$, we have $\fdrop(w_r)=0$ and hence $w_r \in D_1$. This
  is in contradiction to $R\cap D_1=\emptyset$.
  \qed
\end{proof}


\paragraph{Deciding separability}
The next step in our approach is to decide whether a given VASS
language $L$ is contained in $M_k\cup B_\ell\cup \rev{\bar{B}_m}$ for
some $k,\ell,m\in\N$.  Of course this is the case if and only if
$L\subseteq M_k\cup B_k\cup \rev{\bar{B}_k}$ for some $k\in\N$.
Thus, \cref{basic-separators-1vass} essentially tells us that whether
$\sep{L}{D_1}$ holds only depends on three numbers associated to each
word from $L$.  Consider the function $\sigma\colon \Sigma_1^*\to\N^3$
with $\sigma(w)=(\mu(w),\varphi(w),\mu(\rev{\bar{w}}))$. We call a
subset $S\subseteq\N^3$ \emph{separable} if there is a $k\in\N$ so
that for every $(x_1,x_2,x_3) \in S$, we have $x_1\le k$ or $x_3\le k$
or $x_2\not\equiv 0\bmod{k}$. Then, \cref{basic-separators-1vass} can
be formulated as:
\begin{lemma}
  Let $L\subseteq\Sigma_1^*$.  If $L\cap D_1=\emptyset$, then
  $\sep{L}{D_1}$ if and only if $\sigma(L)$ is separable.
\end{lemma}
This enables us to transform $L$ into a
bounded language $\hat{L}$ that behaves the same
in terms of separability from $D_1$. Let
\begin{multline*}
  \hat{L}=\{a_1^m\bar{a}_1^{m+1} a_1^r\bar{a}_1^s a_1^{n+1}\bar{a}_1^{n} \mid \exists w\in L \colon \\
  m\le\mu(w),~~n\le\mu(\rev{\bar{w}}),~~r-s=\varphi(w) \}.
\end{multline*}
Note that if $v=a_1^m\bar{a}_1^{m+1} a_1^r\bar{a}_1^s a_1^{n+1}\bar{a}_1^{n}$,
then we have $\mu(v)=m$ and $\mu(\rev{\bar{v}})=n$ and $\varphi(v)=r-s$. Therefore, the set
$\sigma(\hat{L})$ is separable if and only if $\sigma(L)$ is
separable. Hence, we have:
\begin{lemma}
  For every $L\subseteq\Sigma_1^*$ with $L\cap D_1=\emptyset$, we have
  $\sep{L}{D_1}$ if and only if $\sep{\hat{L}}{D_1}$.
\end{lemma}

Using standard VASS constructions, we can turn $L$ into $\hat{L}$.
\begin{restatable}{lemma}{boundedReplacement}\label{bounded-replacement}
  Given a VASS language $L\subseteq\Sigma_1^*$, one can construct a
  VASS for $\hat{L}$.
\end{restatable}

Before proving~\cref{bounded-replacement} we show how to use it to finalise the argument.
We need to decide whether $\sep{\hat{L}}{D_1}$.
Since $\hat{L}\subseteq B$ with
$B=a_1^*\bar{a}_1^*a_1^*\bar{a}_1^*a_1^*\bar{a}_1^*$, we have
$\sep{\hat{L}}{D_1}$ if and only if $\sep{\hat{L}}{(D_1\cap B)}$. As
subsets of $B$, both $\hat{L}$ and $D_1\cap B$ are bounded languages
and we can decide whether $\sep{\hat{L}}{(D_1\cap B)}$ using
\cref{separability-bounded}.

\medskip

To prove \cref{bounded-replacement}, it is convenient to have a
notion of subsets of $\Sigma^*\times\N^m$ described
by vector addition systems.  First, a \emph{vector addition system}
(VAS) is a VASS that has only one state.  Since it has only one state,
it is not mentioned in the configurations or the transitions.  We say
that $R\subseteq\Sigma^*\times\N^m$ is a \emph{VAS relation} if there
is a $d+m$-dimensional VAS $V$ and vectors $\vec s,\vec t\in\N^d$ such
that
$R=\{(w,\vec u)\in\Sigma^*\times\N^m \mid (\vec s,0)\trans{w}(\vec
t,\vec u)\}$.  Here, $\vec s$ and $\vec t$ are called \emph{source}
and \emph{target vector}, respectively.

However, sometimes it is easier to describe a relation by a VASS than
by a VAS.  We say that $R\subseteq\Sigma^*\times\N^m$ is
\emph{described by the $d+m$-dimensional VASS} $V=(Q,T,s,t,h)$ if
$R=\{(w,\vec u)\in\Sigma^*\times\N^m \mid (s,0,0)\trans{w}(t,0,\vec
u)\}$. Of course, a relation is a VAS relation if and only if it is
described by some VASS and these descriptions are easily translated.

\begin{lemma}\label{vas-relation-product}
  If $R\subseteq\Sigma^*\times\N^m$ and $S\subseteq\Sigma^*\times\N^n$
  are VAS relations, then so is the relation
  $R\oplus S:=\{(w,\vec u,\vec v) \mid (w,\vec u)\in R\wedge (w,\vec v)\in S\}$.
\end{lemma}
\begin{proof}
  We employ a simple product construction. Suppose $V_0$ describes $R$
  and $V_1$ describes $S$.  Without loss of generality, let $V_0$ and
  $V_1$ be $d+m$-dimensional and $d+n$-dimensional, respectively. The
  new VAS $V$ is $2d+m+n$-dimensional and has three types of
  transitions: First, for any letter $a\in\Sigma$, every transition
  $(\vec u_0,\vec v_0)\in\Z^{d+m}$ of $V_0$ with label $a$ and
  $\vec u_0\in\Z^d$ and $\vec v_0\in\Z^m$, every transition
  $(\vec u_1,\vec v_1)\in\Z^{d+n}$ of $V_1$ with label $a$ and
  $\vec u_1\in\Z^d$ and $\vec v_1\in\Z^n$, $V$ has a transition
  $(\vec u_0,\vec u_1, \vec v_0, \vec v_1)\in\Z^{2d+m+n}$ with label
  $a$.

  Second, for every transition $(\vec u,\vec v)\in\Z^{d+m}$ from $V_0$
  labeled $\varepsilon$ with $\vec u\in\Z^d$ and $\vec v\in\Z^m$, $V$
  has an $\varepsilon$-labeled transition
  $(\vec u, 0^d, \vec v, 0^n)\in\Z^{2d+m+n}$. Here, in slight abuse of
  notation, $0^k$ is meant to be a vector of zeros that occupies
  $k$ coordinates. Third, for every transition
  $(\vec u,\vec v)\in\Z^{d+n}$ labeled $\varepsilon$ from $V_1$ with
  $\vec u\in\Z^d$ and $\vec v\in\Z^m$, $V$ has an
  $\varepsilon$-labeled transition $(0^d,\vec u,0^m,\vec v)$.  If
  $\vec s_i$ and $\vec t_i$ are start and target vector of $V_i$ for
  $i\in\{0,1\}$, then $\vec s=(\vec s_0,\vec s_1)$ and
  $\vec t=(\vec t_0,\vec t_1)$ are used as start and target vectors
  for $V$.  Then, it is routine to check that indeed
  $\{(w,\vec u,\vec v) \mid \vec u\in\Z^m,\vec v\in\Z^n,
  (\vec s,\textbf{0})\trans{w}(\vec t,\vec u,\vec v)\}=R\oplus S$.
  \qed
\end{proof}

\begin{lemma}\label{vas-relation-result}
  Given a VAS language $L\subseteq\Sigma^*$ and a VAS relation
  $R\subseteq\Sigma^*\times\N^m$ one can construct a VAS for the language
  $\{a_{1}^{x_{1}}\cdots a_{m}^{x_{m}} \mid \exists w\in L\colon  R(w,x_{1},\ldots,x_{m}) \}$.
\end{lemma}
\begin{proof}
  Suppose $V$ is a $d$-dimensional VAS accepting $L$ and $V'$ is a
  $d+m$-dimensional VAS for $R$. We construct the $2d+m$-dimensional
  VAS $V''$, which has four types of transitions. First, for every
  transition $\vec u\in\Z^d$ labeled $a\in\Sigma$ and every
  $a$-labeled transition $\vec v\in\Z^{d+m}$ in $V'$, we have an
  $\varepsilon$ labeled transition $(\vec u,\vec v)$ in $V''$.
  Second, for every $\varepsilon$-labeled transition $\vec u\in\Z^d$
  in $V$, we have an $\varepsilon$-labeled transition
  $(\vec u,0^{d+m})\in\Z^{2d+m}$ in $V''$. Third, for every
  $\varepsilon$-labeled transition $\vec u\in\Z^{d+m}$ in $V'$, we
  have an $\varepsilon$-labeled transition $(0^d,\vec u)\in\Z^{2d+m}$
  transition in $V''$. Fourth, for every $i\in[1,m]$, we have an
  $a_i$-labeled transition $(0^{2d},-\vec e_i)\in\Z^{2d+m}$, where
  $\vec e_i\in\Z^m$ is the $m$-dimensional unit vector with $1$ in
  coordinate $i$ and $0$ everywhere else. It is now easy construct a
  VASS $V'''$ with $\Lang{V'''}=\Lang{V''}\cap a_1^*\cdots a_n^*$.
  Then clearly, we have
  $\Lang{V'''}=\{a_1^{x_1}\cdots a_m^{x_m}\mid \exists w\in L\colon
  R(w,x_1,\ldots,x_m)\}$.
\qed
\end{proof}

\begin{proof}[Proof of \cref{bounded-replacement}]
  First, let us show that the following relations are VAS relations:
  \begin{align*}
    R_1&=\{(w,n)\in\Sigma_1^*\times\N \mid n\le \mu(w)\},  \\
    R_2&=\{(w,r,s)\in\Sigma_1^*\times\N^2\mid r-s=\varphi(w)\}, \\
    R_3&=\{(w,n)\in\Sigma_1^*\times\N\mid n\le \mu(\rev{\bar{w}}) \}.\end{align*}
  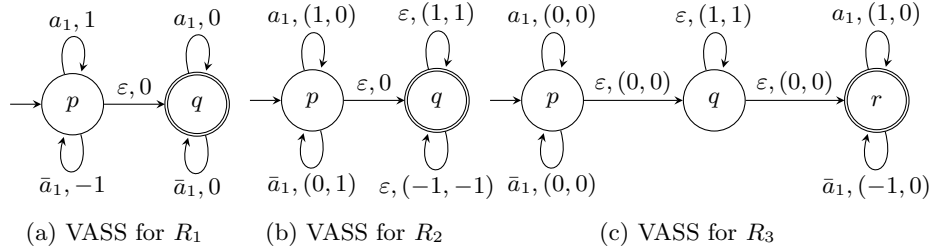
\begin{figure}
    \centering
    {\centering\subcaptionbox{VASS for $R_1$\label{vass-r1}}[0.25\textwidth]{
    \begin{tikzpicture}[->,inner sep=2pt,auto,initial text=,>=stealth,node distance=0.8cm]
        \node[state,initial] (p) {$p$};
        \node[state,accepting] (q) [right=of p] {$q$};
        \path (p) edge [loop above] node {$a_1, 1$} (p);
        \path (p) edge [loop below] node {$\bar{a}_1, -1$} (p);
        \path (p) edge node {$\varepsilon,0$} (q);
        \path (q) edge [loop above] node {$a_1, 0$} (q);
        \path (q) edge [loop below] node {$\bar{a}_1, 0$} (q);
      \end{tikzpicture}
      }}
      {\centering\subcaptionbox{VASS for $R_2$\label{vass-r2}}[0.25\textwidth]{
      \begin{tikzpicture}[->,inner sep=2pt,auto,initial text=,>=stealth,node distance=0.8cm]
        \node[state,initial] (p) {$p$};
        \node[state,accepting] (q) [right=of p] {$q$};
        \path (p) edge [loop above] node {$a_1, (1,0)$} (p);
        \path (p) edge [loop below] node {$\bar{a}_1, (0,1)$} (p);
        \path (p) edge node {$\varepsilon,0$} (q);
        \path (q) edge [loop above] node {$\varepsilon, (1,1)$} (q);
        \path (q) edge [loop below] node {$\varepsilon, (-1,-1)$} (q);
      \end{tikzpicture}
    }}
  {\centering\subcaptionbox{VASS for $R_3$\label{vass-r3}}[0.45\textwidth]{
      \begin{tikzpicture}[->,inner sep=2pt,auto,initial text=,node distance=0.8cm,>=stealth]
        \node[state,initial] (p) {$p$};
        \node[state] (q) [right=1.3cm of p] {$q$};
        \node[state,accepting] (r) [right=1.3cm of q] {$r$};
        \path (p) edge [loop above] node {$a_1, (0,0)$} (p);
        \path (p) edge [loop below] node {$\bar{a}_1, (0,0)$} (p);
        \path (p) edge node {$\varepsilon,(0,0)$} (q);
        \path (q) edge [loop above] node {$\varepsilon, (1,1)$} (q);
        \path (q) edge node {$\varepsilon,(0,0)$} (r);
        \path (r) edge [loop above] node {$a_1, (1,0)$} (r);
        \path (r) edge [loop below] node {$\bar{a}_1, (-1,0)$} (r);
      \end{tikzpicture}
      }}
    \caption{VASS for relations $R_1$, $R_2$, and $R_3$ in the proof of \cref{bounded-replacement}.}
    \label{vas-relations}
  \end{figure}
  In \cref{vass-r1,vass-r2,vass-r3}, we show vector addition systems with states
  for the relations $R_1$, $R_2$, and $R_3$ (it is easy to translate
  them to VAS for the relations). From the VASS for $R_1$ and $R_3$, one can readily build VAS
  for the relations
  \begin{align*} R'_1&=\{(w,m,m+1)\in\Sigma_1^*\times\N^2 \mid m\le\mu(w) \}, \\ R'_3&=\{(w,n+1,n)\in\Sigma_1^*\times\N^2\mid m\le\mu(\rev{\bar{w}})\}.\end{align*}
  
  According to \cref{vas-relation-product}, we can construct a VAS
  for
  $R=R'_1\oplus R_2\oplus R'_3\subseteq\Sigma^*\times\N^6$.
  Applying \cref{vas-relation-result} to $L$ and $R$ yields a VAS for the language
  \begin{multline*} \{ a_1^m a_2^{m+1} a_3^r a_4^s a_5^{n+1} a_6^n \mid \exists w\in L\colon \\
    m\le\mu(w),~r-s=\varphi(w),~n\le\mu(\rev{\bar{w}})\}. \end{multline*}
  Now appropriately renaming the symbols $a_1,\ldots,a_6$ to $a_1$ or $\bar{a}_1$
  yields a VAS for $\hat{L}$.
\qed
\end{proof}


\section{VASS vs. Coverability VASS}\label{vass-cover}

Let us now show \cref{result-vass-cover}.
In~\cite{DBLP:conf/concur/CzerwinskiLMMKS18} it was shown that two
coverability VASS languages $K$ and $L$ are regularly separable if and
only if $K\cap L=\emptyset$ (the result
in~\cite{DBLP:conf/concur/CzerwinskiLMMKS18} applies to all languages
of well-structured transition systems). However, when deciding
$\sep{K}{L}$ for a VASS language $K$ and a coverability VASS language
$L$, a simple disjointness check is not enough: If
$K=\{a^nb^{m}\mid n<m\}$ and $L=\{a^n b^m\mid n\ge m\}$, then $K$ is
in $\cV$ and $L$ is in $\cC$ and we have $K\cap L=\emptyset$, but not
$\sep{K}{L}$.

Instead, we use our approach to reduce separability to the
simultaneous unboundedness
problem~\cite{DBLP:conf/icalp/Zetzsche15,DBLP:journals/dmtcs/CzerwinskiMRZZ17}.
A language $L\subseteq a_1^*\cdots a_n^*$ is \emph{simultaneously
  unbounded} if for every $k\in\N$, there is a word
$a_1^{x_1}\cdots a_n^{x_n}\in L$ with $x_1,\ldots,x_n\ge k$. The
\emph{simultaneous unboundedness problem (SUP)} asks, given a language
$L\subseteq a_1^*\cdots a_n^*$, whether $L$ is simultaneously
unbounded. For VASS languages, this problem is decidable. This follows
from computability of downward
closures~\cite{DBLP:conf/icalp/HabermehlMW10} or from general results
on unboundedness problems for
VASS~\cite{DBLP:conf/icalp/CzerwinskiHZ18}.

\paragraph{Basic separators for coverability VASS} As before, we
develop a notion of basic separators. We start with a version of the
sets $B_{k}$, where we use $C_1$ instead of $D_1$. We set
\[ B'_k=\{w\in\Sigma_1^* \mid \text{$w\notin C_1$ and $\mu(w)\le k$}\}. \]
Just like for $B_k$, the set $B'_k$ is clearly disjoint from
$C_1$.  Moreover, we need variants of these in higher dimension:
For $i\in[1,n]$ and $k\in\N$, let
$B_{i,k}=\lambda_i^{-1}(B'_k)$. This means, $B_{i,k}$ is the set of
walks through $\Z^n$ which have a prefix $p$ such that in coordinate
$i$, the walk $p$ goes below zero, but also never exceeds $k$ before it does so. The sets $B_{i,k}$ are clearly regular,
because an automaton can maintain the $\varphi$-value in coordinate
$i$ of the read prefix in its state as long as it stays positive:
During that time, the value belongs to $[0,k]$.
\begin{restatable}{lemma}{basicSeparatorsCover}\label{basic-separators-cover}
  Let $R\subseteq\Sigma_n^*$ be a regular language. Then
  $R\cap C_n=\emptyset$ if and only if $R$ is included in a finite
  union of sets of the form $B_{i,k}$ for $i\in[1,n]$ and $k\in\N$.
\end{restatable}


In order to show \cref{basic-separators-cover}, we first prove a lemma, for which we need some terminology.
Suppose $R\subseteq\Sigma_n^*$ and for every $k\in\N$, we have
$R\setminus (B_{1,k}\cup\cdots\cup B_{n,k})\ne\emptyset$. We have to
show that then $R\cap C_n\ne\emptyset$.

There is a sequence of words $w_1,w_2,\ldots\in R$ so that for each
$k\in\N$, we have $w_k\notin B_{1,k}\cup\cdots\cup B_{n,k}$. For each
$i\in[1,n]$, the non-membership $w_k\notin B_{i,k}$ is either because
(i)~$\mu(\lambda_i(w))>k$ or (ii)~$\lambda_i(w_k)\in C_1$. By
selecting a subsequence, we may assume that all words agree about
which coordinates satisfy (i) and which satisfy (ii).  Formally, we
may assume that there is a subset $I\subseteq[1,n]$ such that for each
$k\in\N$ and $i\in[1,n]$, we have $\mu(\lambda_i(w_k))>k$ if $i\in I$
and $\lambda_i(w_k)\in C_1$ if $i\notin I$. In this situation, we call
$w_1,w_2,\ldots$ an \emph{$I$-witness sequence}.

\begin{lemma}\label{witness-sequence-reduction}
  Suppose there is an $I$-witness sequence in $R$ for
  $I\subseteq[1,n]$ with $I\ne\emptyset$. Then $R$ has an $I'$-witness
  sequence for some strict subset $I'\subset I$.
\end{lemma}
\begin{proof}
Suppose $w_1,w_2,\ldots$ is an $I$-witness sequence.  Clearly, if
$I=\emptyset$, then $w_k\in C_n$ for every $k\in\N$ and we are
done. So, suppose $I\ne\emptyset$.  We shall prove that there is also
an $I'$-witness sequence with $|I'|<|I|$.

Since for every $i\in I$, we have $\mu(\lambda_i(w_k))>k$ we know that
$w_k$ decomposes as $w_k=u_k^{(i)}v_k^{(i)}$ so that
$\lambda_i(u_k^{(i)})\in C_1$ and
$\varphi(u_k^{(i)})(i)>k$. We pick $i\in I$ so that
$u_k^{(i)}$ has minimal length. We claim that then $u_k^{(i)}\in C_n$.
Indeed, we have $\lambda_j(u_k^{(i)})\in C_1$ for every $j\in[1,n]$:
For $j\in I$, this is because $u_k^{(i)}$ is prefix of some
$u_k^{(t)}$ with $\lambda_j(u_k^{(t)})\in C_1$; for
$j\in [1,n]\setminus I$, this is because $\lambda_j(w_k)\in C_1$.
This proves our claim and thus $u_k^{(i)}\in C_n$.  We define
$u_k=u^{(i)}_k$, $v_k=v^{(i)}_k$, and $m_k:=i$.  By selecting a
subsequence of $w_1,w_2,\ldots$, we may assume that $m_1=m_2=\cdots$
and we define $m:=m_1=m_2=\cdots$.

Consider a finite automaton for $R$. Since each $w_k=u_kv_k$ belongs
to $R$, there is a state $q_k$ so that a run for $w_k$ enters $q_k$
after reading $u_k$. Moreover, let ${\vec u}_k=\varphi(u_k)$. Since we
have $u_k\in C_n$, we know that ${\vec u}_k\in\N^n$.  By selecting a
subsequence of $w_1,w_2,\ldots$, we may assume that $q_1=q_2=\cdots$
and ${\vec u}_1\le {\vec u}_2\le \cdots$ and we define
$q:=q_1=q_2=\cdots$.

We now know that $w_1,w_2,\ldots$ is an $I$-witness sequence and
$m\in I$ and $w_k=u_kv_k$ with $u_k\in C_n$ and
$\varphi(u_k)(m)>k$.  Our goal is to construct an
$I\setminus\{m\}$-witness sequence.  We do this as follows. For each
$k\in\N$, we choose some $\ell$ with $\ell\ge \fdrop(\lambda_m(v_k))$
and $\ell\ge k$. We set $w'_k=u_\ell v_k$ and claim that
$w'_1,w'_2,\ldots$ is an $I\setminus\{m\}$-witness sequence.  First,
note that $w'_k\in R$ for every $k\in\N$. Furthermore, since
${\vec u}_\ell$ is at least $\ell\ge \fdrop(\lambda_m(v_k))$ in
component $m$, we have $\fdrop(\lambda_m(u_\ell v_k))=0$ and hence
$\lambda_m(u_\ell v_k)\in C_1$. Moreover, since
${\vec u}_\ell\ge {\vec u}_k$, we have $\mu(\lambda_i(u_\ell v_k))>k$
for every $i\in I$. Finally, for $i\notin I$, we still have
$\lambda_i(w'_k)\in C_1$ because ${\vec u}_\ell\ge{\vec u}_k$. Thus,
$w'_1,w'_2,\ldots$ is an $I'$-witness sequence for
$I'=I\setminus \{m\}$.
\qed
\end{proof}

We are now prepared to prove \cref{basic-separators-cover}.  Suppose
$R\subseteq\Sigma_n^*$ and for every $k\in\N$, we have
$R\setminus (B_{1,k}\cup\cdots\cup B_{n,k})\ne\emptyset$. As argued
above, this means there is an $I$-witness sequence in $R$. Applying
\cref{witness-sequence-reduction} repeatedly yields an
$\emptyset$-witness sequence in $R$. However, every word in an
$\emptyset$-witness sequence is already a member of $C_n$. Hence,
$R\cap C_n\ne\emptyset$.

\Cref{basic-separators-cover} tells us that to decide whether
$\sep{L}{C_n}$ for a given language $L$, we have to check whether $L$
is included in $B_{1,k}\cup\cdots\cup B_{n,k}$ for some $k\in\N$. Like
in \cref{result-vass-onevass}, we turn $L$ into a different
language. Using standard methods, we can show the following:

\begin{restatable}{lemma}{supReplacement}\label{sup-replacement}
  Given a VASS language $L\subseteq\Sigma_n^*$, one can construct a
  VASS for the language
\[ \hat{L}=\{a_1^{x_1}\cdots a_n^{x_n} \mid \exists w\in L\colon \text{$\mu(\lambda_i(w))\ge x_i$ for $i\in [1,n]$} \}. \]
\end{restatable}

\begin{proof}
  In the proof of \cref{bounded-replacement}, we construct a
  VAS for the relation
  \[ R_1=\{(w,k)\in\Sigma_1^*\times\N \mid k\le \mu(w)\} \]
  (see \cref{vass-r1}). From this, it is easy to obtain a VAS
  for
  \[ S_i = \{(w,k)\in\Sigma_n^* \times\N \mid k\le\mu(\lambda_i(w))
    \} \] for each $i\in[1,n]$.  Indeed, given the VAS for $R_1$, one
  just replaces $a_1$ and $\bar{a}_1$ with $a_i$ and $\bar{a}_i$,
  respectively, and then adds a loops labeled $a_j,0$ and
  $\bar{a}_j,0$ to each state for each $j\in[1,n]$, $j\ne i$. Now,
  using \cref{vas-relation-product}, we build a VAS for the relation
  \begin{multline*}
    S=\{(w,x_1,\ldots,x_m)\in\Sigma_n^*\times\N^m \mid x_i\le\mu(\lambda_i(w)) \\
    \text{for each $i\in[1,m]$}\}.
  \end{multline*}
  Finally, using \cref{vas-relation-result} we can construct a VAS for the language $\hat{L}=\{a_1^{x_1}\cdots a_n^{x_n} \mid x_i\le\mu(\lambda_i(w))~\text{for each $i\in[1,n]$}\}$.
\qed
\end{proof}

Note that $w\in B_{i,k}$ if and only if $\lambda_i(w)\notin C_1$ and
$\mu(\lambda_i(w))\le k$. Therefore, \cref{basic-separators-cover}
implies that $\sep{L}{C_n}$ if and only if $L\cap C_n=\emptyset$ and
$\hat{L}$ is not simultaneously unbounded.

\begin{remark}
  In our decidability proof, we use a polynomial-time Turing reduction
  from regular separability of a VASS language and a coverability VASS
  language to the SUP for VASS languages. There is also such a
  reduction in the converse direction.  This is because given a VASS
  language $L\subseteq a_1^*\cdots a_n^*$, it is easy to construct in
  polynomial time a VASS for
  \[ \tilde{L}=\{a_1^{x_1}\bar{a}_1^{x_1+1}\cdots
    a_n^{x_n}\bar{a}_n^{x_n+1} \mid a_1^{x_1}\cdots a_n^{x_n}\in
    L\}. \] By \cref{basic-separators-cover}, $L$ is
  simultaneously unbounded if and only if
  $\sep{\tilde{L}}{C_n}$. Thus, the problems (i)~regular
  separability of VASS languages and coverability VASS languages and
  (ii)~the SUP for VASS languages are polynomially inter-reducible.
\end{remark}


\section{VASS vs. Integer VASS}\label{vass-zvass}





\begin{figure}
\begin{tikzpicture}[->,>=stealth,initial text=,auto,on grid,scale=1,inner sep=1pt,node distance=1.5cm]
  \node[state,accepting] (q-k) {$q_{-k}$};
  \node[state,draw=none]        (dots) [left= of q-k] {$\cdots$};
  \node[state,accepting] (q-1) [left=of dots] {$q_{-1}$};
  \node[state,accepting,initial] (q0) [left=of q-1] {$q_0$};
  \path (q-k) edge [bend left] node {$a_1$} (dots)
        (dots) edge [bend left] node {$\bar{a}_1$} (q-k);
  \path (q-1) edge [bend left] node {$a_1$} (q0)
        (q0) edge [bend left] node {$\bar{a}_1$} (q-1);
  \path (dots) edge [bend left] node {$a_1$} (q-1)
        (q-1) edge [bend left] node {$\bar{a}_1$} (dots);
  \path (q0) edge [loop above] node {$a_1$} (q0);      
\end{tikzpicture}
\caption{An automaton for $I_k\subseteq\Sigma_1^*$}\label{infix-automaton}
\end{figure}
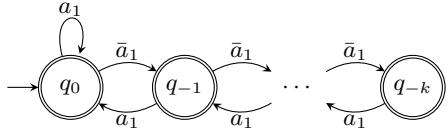

\begin{figure*}
  \begin{tikzpicture}[->,>=stealth,initial text=,auto,on grid,scale=1,inner sep=2pt,node distance=1.5cm]
  \node[state,initial above] (q0) {$q_0$};
  \node[state,accepting] (q1) [right=of q0] {$q_1$};
  \node[state,draw=none] (dots) [right=of q1] {$\cdots$};
  \node[state,accepting] (qk) [right=of dots] {$q_k$};
  \node[state,accepting] (qinf) [right=of qk] {$q_\infty$};
  \node[state,accepting] (q-1) [left=of q0] {$q_{-1}$};
  \node[state,draw=none] (-dots) [left=of q-1] {$\cdots$};
  \node[state,accepting] (q-k) [left=of -dots] {$q_{-k}$};
  \path (q0) edge [bend left] node {$a_1$} (q1)
  (q1) edge [bend left] node {$\bar{a}_1$} (q0);
  \path (q1) edge [bend left] node {$a_1$} (dots)
  (dots) edge [bend left] node {$\bar{a}_1$} (q1);
  \path (dots) edge [bend left] node {$a_1$} (qk)
  (qk) edge [bend left] node {$\bar{a}_1$} (dots);
  \path (qk) edge node {$a_1$} (qinf);
  \path (qinf) edge [loop above] node {$a_1$} (q0');
  \path (qinf) edge [loop below] node {$\bar{a}_1$} (q0');
  \path (q0) edge [bend left] node {$\bar{a}_1$} (q-1)
  (q-1) edge [bend left] node {$a_1$} (q0);
  \path (q-1) edge [bend left] node {$\bar{a}_1$} (-dots)
  (-dots) edge [bend left] node {$a_1$} (q-1);
  \path (-dots) edge [bend left] node {$\bar{a}_1$} (q-k)
  (q-k) edge [bend left] node {$a_1$} (-dots);
\end{tikzpicture}
\caption{Automaton $\cA_k$ with $\Lang{\cA_k}\cap I_k=D_{1,k}$.}\label{direction-automaton}
\end{figure*}
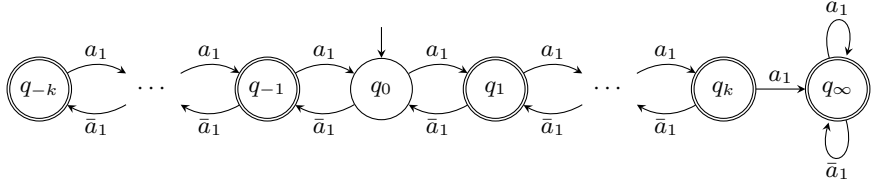

In this section, we apply our approach to solving regular separability
between a VASS language and a $\Z$-VASS language. Here, the collection
of basic separators serves as a geometric characterization of
separability. Proving that it is a set of basic separators is more
involved than in \cref{vass-1vass,vass-cover}.

\subsection{A geometric characterization}\label{geometric}
\Cref{movetrans} tells us that regular separability between a VASS
language and a $\Z$-VASS language amounts to checking whether a given
VASS language $L\subseteq\Sigma_n^*$ is included in some regular
language $R\subseteq\Sigma^*_n$ with $R\cap Z_n=\emptyset$. Therefore,
in this \lcnamecref{geometric}, we classify the regular languages
$R\subseteq\Sigma_n^*$ with $R\cap Z_n=\emptyset$.

A very simple type of such languages is given by modulo counting.  For
$\vec u,\vec v\in\Z^n$, we write $\vec u\equiv \vec v\bmod{k}$ if
$\vec u$ and $\vec v$ are component-wise congruent modulo $k$.
The language
\[ M_k=\{w\in\Sigma_n^* \mid \varphi(w)\not\equiv \textbf{0}\bmod{k}\} \]
is clearly regular and disjoint from $Z_n$.

Since $Z_n$ is commutative (i.e.\ $\Perms(Z_n)=Z_n$), one might expect that it
suffices to consider commutative separators. This is not the case: The
language $L=(a_1\bar{a}_1)^*a_1^+$ is regularly separable from $Z_1$,
but every commutative regular language including $L$ intersects $Z_1$.
Therefore, our second type of regular languages disjoint from $Z_n$ is
non-commutative and we start to describe it in the case
$n=1$. Consider the language
\begin{multline*}
  D_{1,k}=\{w\in\Sigma_1^* \mid \text{$\varphi(w)\ne 0$ and} \\
  \text{for every infix $v$ of $w$: $\varphi(v)\ge -k$}\}.
\end{multline*}
The set $D_{1,k}$ is clearly disjoint from $Z_1$. To see that
$D_{1,k}$ is regular, let us first observe that the language
$I_k=\{w\in\Sigma_1^*\mid \text{for every infix $v$ of $w$:
  $\varphi(v)\ge -k$}\}$ is regular, because the automaton in
\cref{infix-automaton} accepts $I_k$: After reading a word $w$, the
automaton's state reflects the difference $M-\varphi(w)$, where $M$ is
the maximal value $\varphi(v)$ for prefixes $v$ of $w$. Second, the
automaton $\cA_k$ in \cref{direction-automaton} satisfies
$\Lang{\cA_k}\cap I_k=D_{1,k}$: As long as the seen prefix $w$
satisfies $\varphi(w)\in[-k,k]$, the state of $\cA_k$ reflects
$\varphi(w)$ exactly. However, as soon as $\cA_k$ encounters a prefix
$w$ with $\varphi(w)>k$, it enters $q_\infty$. From there, it accepts
every suffix, because an input from $I_k$ can never reach $0$ under
$\varphi$ with such a prefix $w$. Thus, $D_{1,k}$ is regular.

The language $D_{1,k}$ has analogs in higher dimension.
Instead of making sure the value of $\varphi$ never drops more than $k$ along one
particular axis, one can impose this condition in an arbitrary
direction $\vec u\in\Z^n$.
For $\vec u,\vec v\in\Q^n$, $\vec u=(u_1,\ldots,u_n)$,
$\vec v=(v_1,\ldots,v_n)$, we define
$\langle \vec u,\vec v\rangle=u_1v_1+\cdots+u_nv_n$. 
For every vector $\vec u\in \Z^n$ and
$k\in\N\setminus\{0\}$, let
\begin{multline*}
  D_{\vec u,k} = \{w\in\Sigma_n^* \mid \text{$\langle\varphi(w),\vec u\rangle\ne 0$ and } 
  \text{for every infix $v$ of $w$: $\langle\varphi(v),\vec u\rangle\ge -k$}\}.
\end{multline*}
We think of the walks in
$D_{\vec u,k}$ as ``drifting in direction $\vec u$'', hence the name.  To see that
$D_{\vec u,k}$ is regular, consider the morphism
$h_{\vec u}\colon\Sigma_n^*\to\{a_1,\bar{a}_1\}^*$ with
$x\mapsto a_1^{\langle\varphi(x),\vec u\rangle}$ for $x\in\Sigma_n^*$.
Here, we mean $a_1^{\ell}=\bar{a}_1^{|\ell|}$ and
$\bar{a}_1^\ell=a_1^{|\ell|}$ in case $\ell\in\Z$, $\ell<0$. Then we have
$\langle\varphi(w),\vec u\rangle=\varphi(h_{\vec u}(w))$ for any
$w\in\Sigma^*_n$ and hence $D_{\vec u,k}=h_{\vec u}^{-1}(D_{1,k})$.
Therefore, $D_{\vec u,k}$ inherits regularity from $D_{1,k}$.

The main result of this \lcnamecref{geometric} is that the sets $M_k$
and $D_{\vec u,k}$ suffice to explain disjointness of regular
languages from $Z_n$ in the following sense.
\begin{theorem}\label{geometric-disjointness}
  Let $R\subseteq\Sigma^*_n$ be a regular language. Then
  $R\cap Z_n=\emptyset$ if and only if $R$ is included  in a finite
  union of languages of the form $M_k$ and $D_{\vec u,k}$ for $k\in\N$ and
  $\vec u\in\Z^n$.
\end{theorem}
We therefore say that $L\subseteq\Sigma_n^*$ is \emph{geometrically
  separable} if $L$ is contained in a finite union of languages of the form
$M_k$ and $D_{\vec u,k}$.  Then, we can formulate
\cref{geometric-disjointness} as a geometric characterization of
separability from $Z_n$.
\begin{corollary}\label{geometric-separability}
  For $L\subseteq\Sigma^*_n$, we have $\sep{L}{Z_n}$ if and only
  if $L$ is geometrically separable.
\end{corollary}
The rest of \cref{geometric} is devoted to proving
\cref{geometric-disjointness}.

\paragraph{Overview of the proof}
The ``if'' direction of \cref{geometric-disjointness} is clear.  We
show the ``only if'' direction for $L\subseteq\Sigma_n^*$ in several
steps. We first associate to each finite automaton over $\Sigma_n$ a
(rational) cone in $\Q^n$, which must either encompass all of $\Q^n$
or be included in a halfspace (\cref{lem:dichotomy}).

The next step is to decompose the automaton for $L$ into automata
whose strongly connected components form a path.  We then prove
\cref{geometric-disjointness} in the case that such an automaton has
the cone $\Q^n$ (\cref{no-halfspace}).  In the case that the cone of an
automaton $\cA$ is included in some halfspace, we show that
$\Lang{\cA}$ further decomposes into a part inside some $D_{\vec u,k}$
and a part that stays close to some strict linear subspace
$U\subseteq\Q^n$ (\cref{halfspace}).

It thus remains to treat regular languages
$L\subseteq\Sigma_n^*$ whose walks remain close to
$U$.  To this end, we transform $\Lang{\cA}$ into a language in
$\Sigma_m^*$ where $m=\dim U<n$.  The transformation does not affect
disjointness from $Z_n$ (resp. $Z_m$), regularity, or geometric
separability
(\cref{computability-shift,subspace-trans-coordinates,subspace-trans-intersection,subspace-trans-projection,dimension-reduction}).
Since $m<n$, this allows us to apply induction.

\paragraph{Cones of automata}
For a set $S\subseteq\Q^n$, the \emph{cone generated by $S$}
consists of all vectors $x_1\vec u_1+\cdots +x_\ell\vec u_\ell$ where
$x_1,\ldots,x_\ell\in\Q_+$ and $\vec u_1,\ldots,\vec u_\ell\in S$.  To
each automaton $\cA$ over $\Sigma_n$, we associate a cone as follows.
If $w\in\Sigma_n^*$ labels a path in $\cA$, then $\varphi(w)$ is the
\emph{effect} of that path.  Let $\cone(\cA)$ be the cone generated by
the effects of cycles of $\cA$. Since every cycle effect is the sum of
effects of simple cycles, we know that $\cone(\cA)$ is generated by
the effects of simple cycles. In particular, $\cone(\cA)$ is finitely
generated and the set of simple cycle effects can serve as a
representation of $\cone(\cA)$.  A key ingredient in our proof is a
dichotomy of cones (\cref{lem:dichotomy}), which is a
consequence of the well-known Farkas'
lemma~\cite[Corollary~7.1d]{Schrijver1986}. A \emph{half-space} is a
subset of $\Q^n$ of the form
$\{\vec x\in\Q^n\mid \langle \vec x,\vec u\rangle\ge 0\}$ for some
$\vec u\in\Q^n$, $\vec u\ne \textbf{0}$.
\begin{restatable}{lemma}{coneDichotomy}\label{lem:dichotomy}
  For every $\cA$, either $\cone(\cA)=\Q^n$ or $\cone(\cA)$
  is included in some half-space.
\end{restatable}


Let us recall the Farkas' lemma\footnote{The
  formulation of Corollary~7.1d in \cite{Schrijver1986} does not
  specify whether it is over the rationals or the reals. However, on
  p.~85, the author mentions that all results in chapter~7 hold for
  the reals as well as the rationals.}
from linear programming~\cite[Corollary 7.1d]{Schrijver1986}.
Intuitively it states that if a system $A\vec x=\vec b$ of linear inequalities has no solution over $\Q_+$,
then this is certified by a half-space that contains $A\vec x$ for every $\vec x$ over $\Q_+$, but does not contain $\vec b$.
\begin{lemma}[Farkas' Lemma]
  For every $A\in\Q^{n\times m}$ and $\vec b\in\Q^n$, exactly one of the following holds:
  \begin{enumerate}
  \item There exists an $\vec x\in\Q^n$, $\vec x\ge \bf{0}$, with $A\vec x=\vec b$.
  \item There exists a $\vec y\in\Q^n$ with $\vec y^\top A\ge 0$ and
    $\langle \vec y,\vec b\rangle <0$.
  \end{enumerate}
\end{lemma}

\begin{proof}[of \cref{lem:dichotomy}]
  Let $\vec u_1,\ldots,\vec u_k\in\Z^n$ be the effects of all simple
  cycles of $\cA$ and let $C\in\Z^{n\times k}$ be the matrix with
  columns $\vec u_1,\ldots,\vec u_k$. Then $\cone(\cA)$ consists of
  those vectors of the form $C\vec x$ with $\vec x\in \Q_+^n$.

  If $\cone(\cA)\ne\Q^n$, then there is a vector $\vec v\in\Q^n$ with
  $\vec v\notin\cone(\cA)$. This means the system of inequalities
  $C\vec x=\vec v$, $\vec x\ge 0$, does not have a solution. By
  Farkas' lemma, there exists a vector $\vec y\in\Q^n$ with
  $\vec y^\top C\ge 0$ and $\langle \vec y,\vec v\rangle <0$.  Hence,
  for every element $C\vec x$, $\vec x\in\Q_+^n$, of $\cone(\cA)$, we
  have $\langle \vec y,C\vec x\rangle=\vec y^\top C\vec x\ge 0$. Since
  $\vec y\in\Q^n$, there is a $k\in\N$ so that
  $\vec u=k\vec y\in\Z^n$.  Then we have
  $\cone(\cA)\subseteq \{\vec x\in\Q^n \mid \langle \vec x,\vec
  u\rangle\ge 0\}$.
  \qed
\end{proof}

\paragraph{Linear automata} For an automaton $\cA$, 
consider the directed acyclic graph (dag) consisting of strongly
connected components of $\cA$. If this dag is a path, then $\cA$ is
called \emph{linear}. Given an automaton $\cA$, we can
construct linear automata $\cA_1,\ldots,\cA_\ell$ with
$\Lang{\cA}=\Lang{\cA_1}\cup\cdots\cup\Lang{\cA_\ell}$.

\begin{lemma}\label{no-halfspace}
  Let $\cA$ be a linear automaton with $\cone(\cA)=\Q^n$. If
  $\Lang{\cA}\cap Z_n=\emptyset$, then $\Lang{\cA}\subseteq M_k$ for
  some $k$.
\end{lemma}
\begin{proof}
  Since $\cone(\cA)=\Q^n$, we know that in particular the vectors
  $\vec e_1,-\vec e_1,\ldots,\vec e_n,-\vec e_n$ belong to
  $\cone(\cA)$. This means there are cycles labeled $w_1,\ldots,w_p$
  such that both $\vec e_i$ and $-\vec e_i$ are linear combinations of
  $\varphi(w_1),\ldots,\varphi(w_p)$ with coefficients in $\Q_+$, for
  every $i\in[1,n]$.  Therefore, there is a $k\in\N$ such that
  $k\cdot \vec e_i$ and $-k\cdot \vec e_i$ are linear combinations of
  $\varphi(w_1),\ldots,\varphi(w_p)$ with coefficients in $\N$, for
  every $i\in[1,n]$.
  We claim that $\Lang{\cA}\subseteq M_k$. Towards a contradiction,
  suppose $w\in \Lang{\cA}$ with $\varphi(w)\equiv \textbf{0}\bmod{k}$. Since
  $\cA$ is linear, we can take the run for $w$ and insert cycles so
  that the resulting run visits every state in $\cA$. Instead of
  inserting every cycle once, we insert it $k$ times, so that the
  resulting run (i)~visits every state in $\cA$ and (ii)~reads a word
  $w'\in\Sigma_n^*$ with $\varphi(w')\equiv \varphi(w)\bmod{k}$. Now
  since $\varphi(w')\equiv \varphi(w)\equiv \textbf{0}\bmod{k}$, we can write
  $-\varphi(w')=x_1\varphi(w_1)+\cdots+x_p\varphi(w_p)$ with
  coefficients $x_1,\ldots,x_p\in\N$. Since in the run for $w'$, every
  state of $\cA$ is visited, we can insert cycles corresponding to the
  $w_1,\ldots,w_p$: For each $i\in[1,p]$, insert the cycle for $w_i$
  exactly $x_i$ times. Let $w''$ be the word read by the resulting run
  and note that $w''\in\Lang{\cA}$.  Then we have
  $\varphi(w'')=\varphi(w')+x_1\varphi(w_1)+\cdots+x_p\varphi(w_p)=\textbf{0}$
  and thus $w''\in Z_n$, contradicting $\Lang{\cA}\cap Z_n=\emptyset$.
\end{proof}

\paragraph{Walks that stay close to a subspace}
Suppose we are given a vector space $U\subseteq\Q^n$ (represented by a
basis) with $m=\dim U<n$ and a bound $\ell\ge 0$.
Let $\|\vec u\|=\sqrt{\langle \vec u,\vec u\rangle}$.
For $U\subseteq\Q^n$ and $\vec v\in\Q^n$, we set $d(\vec v,U)=\inf\{\|\vec v-\vec x\|\mid \vec x\in U\}$.
Then we define the set
\[ S_{U,\ell}=\{w\in\Sigma_n^* \mid \text{for every prefix $v$ of $w$:
    $d(\varphi(v),U)\le \ell$}\}.\]
Hence, $S_{U,\ell}$ collects those
walks whose prefixes stay close to the subspace $U$. 

\begin{lemma}\label{halfspace}
  Let $\cA$ be an automaton such that $\cone(\cA)$ is contained in
  some half-space.  One can compute $k,\ell\in\N$, $\vec u\in\Z^n$, and a
  strict subspace $U\subseteq \Q^n$ with $\Lang{\cA}\subseteq D_{\vec u,k}\cup
  S_{U,\ell}$.
\end{lemma}

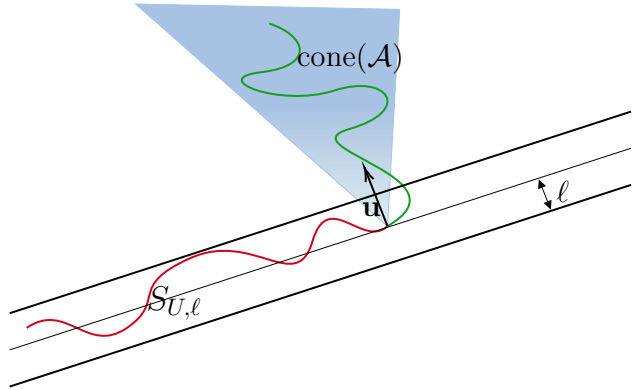
\begin{figure}  
  \centering
  \scalebox{0.5}{
  
\tikzset {_bnuccimim/.code = {\pgfsetadditionalshadetransform{ \pgftransformshift{\pgfpoint{0 bp } { 0 bp }  }  \pgftransformrotate{-270 }  \pgftransformscale{2 }  }}}
\pgfdeclarehorizontalshading{_tf8jpoegn}{150bp}{rgb(0bp)=(0.88,0.93,0.94);
rgb(37.5bp)=(0.88,0.93,0.94);
rgb(50.89285714285714bp)=(0.66,0.76,0.89);
rgb(100bp)=(0.66,0.76,0.89)}
\tikzset{every picture/.style={line width=0.75pt}} 

\begin{tikzpicture}[x=0.75pt,y=0.75pt,yscale=-1,xscale=1]

\path  [shading=_tf8jpoegn,_bnuccimim] (155.6,-6.9) -- (419.19,-1.5) -- (407,216.61) -- cycle ; 
 \draw  [color={rgb, 255:red, 164; green, 197; blue, 235 }  ,draw opacity=1 ] (155.6,-6.9) -- (419.19,-1.5) -- (407,216.61) -- cycle ; 

\draw [color={rgb, 255:red, 208; green, 2; blue, 27 }  ,draw opacity=1 ][line width=1.5]    (47,318.98) .. controls (87,291.31) and (107,346.64) .. (147,318.98) .. controls (187,291.31) and (147,277.48) .. (210,248.89) .. controls (273,220.3) and (316,286.7) .. (332,232.29) .. controls (348,177.88) and (382,239.67) .. (407,216.61) ;

\draw [color={rgb, 255:red, 17; green, 159; blue, 31 }  ,draw opacity=1 ][line width=1.5]    (407,216.61) .. controls (491,167.73) and (305,137.3) .. (369,118.85) .. controls (433,100.41) and (408,62.6) .. (354,80.12) .. controls (300,97.64) and (226,68.13) .. (281,61.68) .. controls (336,55.22) and (326,24.79) .. (289,12.8) ;

\draw [line width=1.5]    (30,304.22) -- (650,101.33) ;

\draw [line width=1.5]    (30,378) -- (650,175.11) ;

\draw [line width=1.5]    (407,216.61) -- (385.09,160.38) ;
\draw [shift={(384,157.59)}, rotate = 428.71000000000004] [color={rgb, 255:red, 0; green, 0; blue, 0 }  ][line width=1.5]    (14.21,-4.28) .. controls (9.04,-1.82) and (4.3,-0.39) .. (0,0) .. controls (4.3,0.39) and (9.04,1.82) .. (14.21,4.28)   ;

\draw    (30,341.11) -- (650,138.22) ;

\draw [line width=0.75]    (558.73,171.44) -- (569.27,198.15) ;
\draw [shift={(570,200.01)}, rotate = 248.48000000000002] [fill={rgb, 255:red, 0; green, 0; blue, 0 }  ][line width=0.75]  [draw opacity=0] (8.93,-4.29) -- (0,0) -- (8.93,4.29) -- cycle    ;
\draw [shift={(558,169.58)}, rotate = 68.48] [fill={rgb, 255:red, 0; green, 0; blue, 0 }  ][line width=0.75]  [draw opacity=0] (8.93,-4.29) -- (0,0) -- (8.93,4.29) -- cycle    ;

\draw (391,199.4) node  [scale=1.6] [align=left] {\textbf{{\Large u}}};
\draw (194,294.08) node [scale=1.8] [align=left]  {\textbf{\Large $S_{U,\ell}$}};
\draw (582,181.57) node [scale=1.6]  [align=left] {\textbf{{\Large $\ell$}}};
\draw (370,47.6) node [scale=1.6] [align=left] {\textbf{{\Large $\cone(\cA)$}}};

\end{tikzpicture}}
  \caption{Two runs (red and green) inside $D_{\vec u,k}\cup
  S_{U,\ell}$.}
\end{figure}

\begin{proof}
  Suppose $\cone(A)\subseteq H$, where
  $H=\{\vec v\in \Q^n \mid \langle \vec v,\vec u\rangle\ge 0\}$ for
  some vector $\vec u\in\Q^n\setminus \{\textbf{0}\}$. Without loss of
  generality, we may assume $\vec u\in\Z^n\setminus\{\textbf{0}\}$. Let
  $U=\{\vec v\in\Q^n \mid \langle \vec v,\vec u\rangle=0\}$. Then
  clearly $\dim U=n-1$.  Observe that since $\cone(\cA)\subseteq H$, we
  have $\langle \vec v,\vec u\rangle\ge 0$ for every cycle effect
  $\vec v\in\Z^n$ of $\cA$.  Let $k$ be the number of states in $\cA$. Now,
  whenever $w\in\Lang{\cA}$ and $v$ is an infix of $w$, then
  $\langle \varphi(v),\vec u\rangle\ge -k$: If
  $\langle \varphi(v),\vec u\rangle<-k$, then the path reading $v$
  must contain a cycle reading $v'\in\Sigma^*_n$ with
  $\langle \varphi(v'),\vec u\rangle<0$, which contradicts
  $\cone(\cA)\subseteq H$.

  We claim that $\Lang{\cA}\subseteq D_{\vec u,k}\cup S_{U,k}$. Let
  $w\in\Lang{\cA}$.  We distinguish two cases.
%
  \emph{Case~1:}~Suppose $w$ has a prefix $v$ with
    $\langle\varphi(v),\vec u\rangle>k$.  Write $w=vv'$. As argued
    above, we have $\langle \varphi(v'),\vec u\rangle\ge -k$.
    Hence,
    $\langle \varphi(w),\vec u\rangle=\langle \varphi(v),\vec
    u\rangle+\langle\varphi(v'),\vec u\rangle >0$. Thus, we have
    $w\in D_{\vec u,k}$.
%
    \emph{Case~2:}~Suppose for every prefix $v$ of $w$, we have
    $\langle \varphi(v),\vec u\rangle\le k$.  Then, for every prefix
    $v$ of $w$, we have $-k\le\langle \varphi(v),\vec u\rangle\le k$
    and thus 
    $d(\varphi(v),U)=|\langle \varphi(v), \vec u\rangle|/\|\vec u\|\le
    k$ (see \cref{distance-hyperplane}).  Thus, $w\in S_{U,k}$.
\end{proof}

\begin{lemma}\label{distance-hyperplane}
  Let $\vec u\in\Q^n$ and
  $U=\{\vec v\in\Q^n \mid \langle \vec v,\vec u\rangle=0\}$.  Then
  $d(\vec v,U)=\frac{|\langle \vec v,\vec u\rangle|}{\|\vec u\|}$ for
  $\vec v\in\Q^n$.
\end{lemma}
\begin{proof}
  We extend $\vec u$ to an orthogonal basis $\vec b_1,\ldots,\vec b_n$
  of $\Q^n$, meaning $\langle \vec b_i,\vec b_j\rangle=0$ if $i\ne j$ and
  $\vec b_1=\vec u$. Because of orthogonality, we may express $\|\vec v\|$ for any vector
  $\vec v\in\Q^n$ with $\vec v=v_1\vec b_1+\cdots+v_n\vec b_n$ as
  \begin{multline*}
    \sqrt{\langle \vec v,\vec v\rangle}=\sqrt{\langle v_1\vec b_1+\cdots +v_n\vec b_n,v_1\vec b_1+\cdots +v_n\vec b_n\rangle} \\
    =\sqrt{\sum_{i=1}^n v_i^2\langle \vec b_i,\vec b_i\rangle}=\sqrt{\sum_{i=1}^n v_i^2\|\vec b_i\|^2}.
  \end{multline*}
  Let $\vec x\in\Q^n$ be a vector with
  $\vec x=x_1\vec b_1+\cdots+x_n\vec b_n$.  Since $\vec b_1=\vec u$
  and thus $\langle \vec x,\vec u\rangle=x_1$, the vector $\vec x$
  belongs to $U$ if and only if $x_1=0$. Therefore, for
  $\vec v\in \Q^n$ with $\vec v=v_1\vec b_1+\cdots+v_n\vec b_n$ and $\vec x\in U$, we
  have
  \[ \|\vec v-\vec x\|=\sqrt{v_1^2\|\vec u\|^2+\sum_{i=2}^n (v_i-x_i)^2\|\vec b_i\|^2}. \]
   This distance is minimal with $x_i=v_i$ for $i\in[2,n]$ and
  in that case, the distance is
  $d(\vec v,U)=\sqrt{v_1^2\|\vec u\|^2}=|v_1|\cdot \|\vec u\|$. Since
  $\langle \vec v,\vec u\rangle=v_1\|\vec u\|^2$, that implies
  $d(\vec v,U)=|\langle \vec v,\vec u\rangle|/\|\vec u\|$.
  \qed
\end{proof}

\paragraph{Mapping to lower dimension}
\Cref{halfspace} tells us that if $\cone(\cA)$ is included in some
halfspace, then $\Lang{\cA}$ can be split into (i)~a part
$L\cap D_{\vec u,k}$ that is already geometrically separable and
(ii)~a part $L\cap S_{U,\ell}$ that stays close to a strict subspace
$U\subseteq\Q^n$. Therefore, to complete the proof that regular
languages disjoint from $Z_n$ are geometrically separable, it remains
to treat subsets of $S_{U,\ell}$. We will now show that they can be
transformed into a language in $\Sigma_m^*$, where $m=\dim U<n$.  This
transformation will not affect regularity, geometric separability or disjointness
from $Z_n$ (resp. $Z_m$) and thus allow us to apply induction.

This transformation will be performed by a transducer.
The transducer will consist of three steps, \emph{coordinate
  transformation} ($f$), \emph{intersection} ($R_{V,p}$), and
\emph{projection} ($\pi_m$). In the coordinate transformation, we
translate $L$ from walks that stay close to $U$ into walks that stay
close to $V=\{(v_1,\ldots,v_n)\in \Q^n\mid v_{m+1}=\cdots=v_n=0\}$.
The intersection will then select only those walks that not only stay
close to $V$, but even arrive in $V$. Finally, we project away
the coordinates $m+1,\ldots,n$ and thus have walks in $\Z^m$.

We now describe each of the three steps.  For the coordinate
transformation, we apply a linear map $A$ to the walk in $L$ that maps
$U$ to $V$. Let us define this map as a matrix $A\in\Z^{n\times
  n}$. We choose an orthogonal basis
$\vec b_1, \ldots,\vec b_n\in\Z^n$ of $\Q^n$ such that
$\vec b_1,\ldots,\vec b_m$ is a basis for $U$. This can be done,
e.g. using Gram-Schmidt
orthogonalisation~\cite{LangLinearAlgebra1966}.  If
$B\in\Z^{n\times n}$ is the matrix whose columns are
$\vec b_1,\ldots,\vec b_n$, then $B$ is invertible and maps $V$ to
$U$.  Thus the inverse $B^{-1}\in\Q^{n}$ maps $U$ to $V$. We can
clearly choose an $\alpha\in\Z$ such that
$\alpha B^{-1}\in\Z^{n\times n}$ and we set $A=\alpha B^{-1}$.
For each $i\in[1,n]$, choose a word $w_i\in\Sigma_n^*$ with
$\varphi(w_i)=A\varphi(a_i)$ and let $f\colon\Sigma^*_n\to\Sigma^*_n$
be the morphism with $f(a_i)=w_i$ and $f(\bar{a}_i)=\bar{w}_i$.
Now $f$ indeed transforms walks close to $U$ into walks close to $V$:
\begin{restatable}{lemma}{computabilityShift}\label{computability-shift}
  We can compute $p\in\N$ with $f(S_{U,\ell})\subseteq S_{V,p}$.
\end{restatable}
\begin{proof}
  Choose $k\in\N$ so that $k\ge |f(a_i)|$ and $k\ge |f(\bar{a}_i)|$
  for $i\in[1,n]$ and let $p=\|A\|\cdot\ell+k$. We claim that
  $f(S_{U,\ell})\subseteq S_{V,p}$. Let $w\in S_{U,\ell}$.
  
  Consider a prefix $v$ of $f(w)$. Let us first consider the case that
  $v=f(u)$ for some prefix $u$ of $w$. Since $w\in S_{U,\ell}$,
  we have $d(\varphi(u), U)\le\ell$. Therefore,
  \begin{align*}
    d(\varphi(f(u)),V)&=d(A\varphi(u), AU) \\
    &=\inf\{\|A\varphi(u)-A\vec u\| \mid \vec u\in U\} \\
                      &\le \|A\| \cdot \inf\{\|\varphi(u)-\vec u\| \mid \vec u\in U\}\\
    &=\|A\|\cdot d(\varphi(u),U)=\|A\|\cdot \ell.
  \end{align*}
  Now if $v$ is any prefix of $f(w)$, then $v=f(u)v'$, where $u$ is a
  prefix of $w$ and $|v'|\le k$. This implies that
  $d(\varphi(v),V)\le d(\varphi(u),V)+k\le \|A\|\cdot \ell+k=p$.
  \qed
\end{proof}

Moreover, applying $f$ does not introduce geometric separability
and preserves disjointness with $Z_n$.
\begin{restatable}{lemma}{subspaceTransCoordinates}\label{subspace-trans-coordinates}
  If $f(L)$ is geometrically separable for $L\subseteq\Sigma^*_n$, then so is $L$.
  We have $L\cap Z_n=\emptyset$ if and only if $f(L)\cap Z_n=\emptyset$.
\end{restatable}


\begin{proof}
  For the first statement, we prove that
  $f^{-1}(M_k)\subseteq M_k$ and
  $f^{-1}(D_{\vec u,k})\subseteq D_{A^\top \vec u,k}$, which clearly
  suffices. Note that if $w\in\Sigma_n^*$ satisfies
  $\varphi(w)\equiv \textbf{0}\bmod{k}$, then also
  $\varphi(f(w))=A \varphi(w)\equiv \textbf{0}\bmod{k}$. This implies
  $f^{-1}(M_k)\subseteq M_k$.  For the second inclusion, suppose
  $w\in\Sigma_n^*$ satisfies $f(w)\in D_{\vec u,k}$ and let $v$ be a prefix
  of of $w$.  Then $f(v)$ is a prefix of $f(w)$ and thus
  \begin{align*}
    \langle \varphi(v),A^\top \vec u\rangle &= \varphi(v)^\top A^\top \vec u = (A\varphi(v))^\top \vec u \\
    &=\langle A\varphi(v),\vec u\rangle=\langle \varphi(f(v)), \vec u\rangle.
  \end{align*}
  In particular, we have
  $\langle\varphi(v),A^\top\vec u\rangle=\langle \varphi(f(v)),\vec
  u\rangle\ge -k$ and
  $\langle\varphi(w),A^\top \vec u\rangle=\langle\varphi(f(w)),\vec
  u\rangle>0$, which implies $w\in D_{A^\top \vec u,k}$.

  For the second statement, note that $A$ is invertible, meaning
  $\varphi(f(w))=A\varphi(w)$ vanishes if and only if $\varphi(w)$
  vanishes.
  \qed
\end{proof}

For the second step of our transformation (intersection), we observe
the following.  Since the walks in $S_{V,p}$ stay close to $V$, there
is a finite set $F\subseteq\Z^{n-m}$ of possible difference vectors
between a point $\varphi(v)$ reached by a prefix $v$ of a word in
$S_{V,p}$ and the point closest to $\varphi(v)$ in $V$. Therefore, the
set $R_{V,p}$ of walks in $S_{V,p}$ that also arrive in $V$ is
regular: One can maintain the current distance vector in the state.
To make this formal, let $\bar{\pi}_j\colon\Q^n\to\Q^j$ denote the
projection on the last $j$ coordinates,
$\bar{\pi}_j(v_1,\ldots,v_n)=(v_{n-j+1},\ldots,v_n)$.  Then we have
$d(\vec v, V)=\|\bar{\pi}_{n-m}(\vec v)\|$ for every $\vec v\in\Q^n$.
If $v$ is a prefix of $w\in S_{V,p}$, then $d(\varphi(v),V)\le p$
implies $\|\bar{\pi}_{n-m}(\varphi(v))\|\le p$ and hence there is a
finite set $F\subseteq\Z^{n-m}$ such that
$\bar{\pi}_{n-m}(\varphi(v))\in F$ for every prefix $v$ of some
$w\in S_{V,p}$.  Thus, the set
$R_{V,p}=\{w\in S_{V,p} \mid \varphi(w)\in V \}$ is regular. The
second step of our transformation is to intersect with $R_{V,p}$.
\begin{restatable}{lemma}{subspaceTransIntersection}\label{subspace-trans-intersection}
  Let $L\subseteq S_{V,p}$. If $L\cap R_{V,p}$ is geometrically separable, then so is $L$.
  We have $L\cap Z_n=\emptyset$ if and only if $(L\cap R_{V,p})\cap Z_n=\emptyset$.
\end{restatable}


\begin{proof}
  Suppose $L\cap R_{V,p}$ is geometrically separable.  Let
  $\hat{R}_{V,p}=\{w\in S_{V,p} \mid \varphi(w)\notin V\}$. Then
  $S_{V,p}=\hat{R}_{V,p}\cup R_{V,p}$. It suffices to show that
  $\hat{R}_{V,p}\subseteq M_k$ for some $k\in\N$, because then
  \begin{equation} L=(L\cap \hat{R}_{V,p})\cup (L\cap R_{V,p})\subseteq M_k\cup (L\cap R_{V,p}) \label{decomp-zero-nonzero} \end{equation}
  and $L\cap R_{V,p}$ being geometrically separable implies that $L$
  is geometrically separable as well.

  To show that $\hat{R}_{V,p}\subseteq M_k$, let $F\subseteq\Z^{n-m}$
  be a finite set such that $\bar{\pi}_{n-m}(\varphi(v))\in F$ for
  every prefix $v$ of a word $w\in S_{V,p}$. Moreover, choose $k\in\N$
  so that $k>\|\vec v\|$ for every $\vec v\in F$. We claim that then
  $\hat{R}_{V,p}\subseteq M_k$.  To this end, suppose
  $w\in \hat{R}_{V,p}$.  Then $d(\varphi(w),V)\ne 0$ and hence
  $\bar{\pi}_{n-m}(\varphi(w))\in F\setminus \{\textbf{0}\}$.  In particular,
  we have $\varphi(w)\not\equiv \textbf{0}\bmod{k}$ and thus $w\in M_k$. This
  proves $\hat{R}_{V,p}\subseteq M_k$.

  For the second statement, note that $L\cap Z_n=\emptyset$ clearly
  implies $(L\cap R_{V,p})\cap Z_n=\emptyset$. Conversely, if
  $(L\cap R_{V,p})\cap Z_n=\emptyset$, then \cref{decomp-zero-nonzero}
  entails $L\cap Z_n\subseteq (L\cap R_{V,p})\cap Z_n=\emptyset$
  because $M_k\cap Z_n=\emptyset$.
  \qed
\end{proof}

In our third step, we project onto the first $m$ coordinates: We
define $\pi_m\colon \Sigma_n^*\to\Sigma_m^*$ as the morphism with
$\pi_m(a_i)=a_i$, $\pi_m(\bar{a}_i)=\bar{a}_i$ for $i\in[1,m]$, and
$\pi_m(a_i)=\pi_m(\bar{a}_i)=\varepsilon$ for $i\in[m+1,n]$. In other words,
$\pi_m$ deletes the letters $a_i$ and $\bar{a}_i$ for $i\in[m+1,n]\}$.


\begin{restatable}{lemma}{subspaceTransProjection}\label{subspace-trans-projection}
  Let $L\subseteq R_{V,p}$. If $\pi_m(L)$ is geometrically separable,
  then so is $L$.  Moreover, $L\cap Z_n=\emptyset$ if and only if
  $\pi_m(L)\cap Z_n=\emptyset$.
\end{restatable}

\begin{proof}
  It suffices to show that for $w\in R_{V,p}$, two implications hold:
  (i)~if $\pi_m(w)\in M_k$ for some $k\in\N$, then $w\in M_k$ and
  (ii)~if $\pi_m(w)\in D_{\vec u,k}$ for some $\vec u\in\Z^m$ and
  $k\in\N$, then $w\in D_{\vec u',k}$ for some $\vec u'\in\Z^n$.

  Suppose $w\in R_{V,p}$ and $\pi_m(w)\in M_k$. Since $w\in R_{V,p}$,
  the last $n-m$ components of $\varphi(w)$ are zero. Thus, we have
  $\varphi(w)\equiv \textbf{0}\bmod{k}$ if and only if
  $\varphi(\pi_m(w))\equiv \textbf{0}\bmod{k}$. This implies $w\in M_k$.

  Now suppose $w\in R_{V,p}$ with $\pi_m(w)\in D_{\vec u,k}$ for some
  $\vec u\in\Z^m$ and $k\in\N$. Let $\vec u=(u_1,\ldots,u_m)$ and
  define $\vec u'=(u_1,\ldots,u_m,0,\ldots,0)\in\Z^n$.  Then clearly
  \[ \langle \varphi(v),\vec u'\rangle=\langle \varphi(\pi_m(v)), \vec
    u\rangle \]
  for every word $v\in\Sigma^*_n$.  In particular, we have
  $w\in D_{\vec u',k}$.
  \qed
\end{proof}

We are now prepared to define our transformation: Let
$T_{U,\ell}\subseteq\Sigma_n^*\times\Sigma_m^*$ be the transduction
with $T_{U,\ell}L=\pi_m(f(L)\cap R_{V,p})$. Then
\cref{subspace-trans-coordinates,subspace-trans-intersection,subspace-trans-projection} imply:
\begin{restatable}{proposition}{dimensionReduction}\label{dimension-reduction}
  Let $L\subseteq S_{U,\ell}$. If $T_{U,\ell}L\subseteq \Sigma_m^*$ is
  geometrically separable, then so is $L$.  Also,
  $L\cap Z_n=\emptyset$ if and only if
  $(T_{U,\ell}L)\cap Z_m=\emptyset$. Thus, $\sep{L}{Z_n}$ if
  and only if $\sep{(T_{U,\ell}L)}{Z_m}$.
\end{restatable}
\begin{proof}
  With \cref{computability-shift}, the first two statements of
  \cref{dimension-reduction} follow directly from
  \cref{subspace-trans-coordinates,subspace-trans-intersection,subspace-trans-projection}.
  Let us prove the conclusion in the second statement.

  If $\sep{L}{Z_n}$ with a regular $R$ with $L\subseteq R$ and
  $R\cap Z_n=\emptyset$, then by \cref{dimension-reduction}, we have
  $T_{U,\ell}R\cap Z_m=\emptyset$. Hence, $T_{U,\ell}R$ separates
  $T_{U,\ell}L$ and $Z_m$.  Conversely, if $\sep{T_{U,\ell}L}{Z_m}$,
  then by \cref{geometric-separability}, the language $T_{U,\ell}L$ is
  geometrically separable. According to \cref{dimension-reduction},
  that implies that $L$ is geometrically separable and in particular
  $\sep{L}{Z_n}$.
  \qed
\end{proof}

Let us now prove \cref{geometric-disjointness}.
  Suppose $R\subseteq\Sigma_n^*$ and
  $R\cap Z_n=\emptyset$. We show by induction on the
  dimension $n$ that then, $R$ is included in a
  finite union of sets of the form $M_k$ and $D_{\vec u,k}$.
  Let $R=\Lang{\cA}$ for an automaton $\cA$. Since $\cA$ can be
  decomposed into a finite union of linear automata, it suffices to
  prove the claim in the case that $\cA$ is linear. If
  $\cone(\cA)=\Q^n$, then \cref{no-halfspace} tells us that
  $\Lang{\cA}\subseteq M_k$ for some $k\in\N$.  If $\cone(\cA)$ is
  contained in some half-space, then according to \cref{halfspace}, we
  have $R\subseteq D_{\vec u,k}\cup S_{U,\ell}$ for some
  $\vec u\in \Q^n\setminus\{\textbf{0}\}$, $k,\ell\in\N$, and strict subspace
  $U\subseteq\Q^n$. This implies that the regular language
  $R\setminus D_{\vec u,k}$ is included in $S_{U,\ell}$.  We may
  therefore apply \cref{dimension-reduction}, which yields
  $T_{U,\ell}(R\setminus D_{\vec u,k})\cap Z_m=\emptyset$. Since
  $T_{U,\ell}(R\cap S_{U,\ell})\subseteq\Sigma_m^*$ with $m=\dim U<n$,
  induction tells us that $T_{U,\ell} (R\setminus D_{\vec u,k})$ is
  geometrically separable and hence, by \cref{dimension-reduction}, $R\setminus D_{\vec u,k}$ is geometrically
  separable. Since
  $R\subseteq D_{\vec u,k}\cup (R\setminus D_{\vec u,k})$, $R$ is geometrically separable.

\subsection{The decision procedure}\label{vass}
In this \lcnamecref{vass}, we apply \cref{geometric-separability} to
prove \cref{result-vass-zvass}.

Before we prove \cref{result-vass-zvass} let us explain why
a particular straightforward approach does not work.
\Cref{geometric-separability} tells us that
in order to decide whether $\sep{L}{Z_n}$ for $L\subseteq\Sigma_n^*$,
it suffices to check whether there are $k\in\N$,
$\ell_1,\ldots,\ell_m\in\N$, and vectors
$\vec u_1,\ldots,\vec u_m\in\Z^n$ such that
$L\subseteq M_k\cup D_{\vec u_1,\ell_1}\cup\cdots\cup D_{\vec
  u_m,\ell_m}$.
It is tempting to conjecture that there is a finite collection of direction vectors
$F\subseteq\Z^n$ (such as a basis together with negations)
so that for a given language $L$, such an inclusion holds only if it
holds with some $\vec u_1,\ldots,\vec u_m\in F$.
In that case we would
only need to consider scalar products of words in $L$ with vectors in
$F$ and thus reformulate the problem over sections of reachability
sets of VASS. However, this is not the case. For
$\vec u,\vec v\in\Q^n$, we write $\vec u\sim\vec v$ if
$\Q_+\vec u=\Q_+\vec v$. Since $\sim$ has infinitely many equivalence
classes and every class intersects $\Z^n$, the following shows that
there is no  fixed set of directions.

\begin{restatable}{proposition}{noFixedDirections}\label{no-fixed-directions}
  For each $\vec u\in\Z^n$, there is a $k_0\in\N$ such that for
  $k\ge k_0$, the following holds.  For every
  $\ell,\ell_1\ldots,\ell_m\ge 1$, $\vec u_1,\ldots,\vec u_m\in\Z^n$
  with $\vec u_i\not\sim\vec u$ for $i\in[1,n]$, we have
  $D_{\vec u,k}\not\subseteq M_{\ell}\cup D_{\vec u_1,\ell_1}\cup\cdots
  D_{\vec u_m,\ell_m}$.
\end{restatable}
\begin{proof}
  The idea of the proof is as follows. We construct a word $w$
  corresponding to a walk in $\Z^n$, which traverses long distances in
  many directions orthogonal to $\vec u$.  That way, $w$ cannot belong
  to any of the languages $D_{\vec u_i, \ell_i}$ for $i\in[1,m]$.
  Moreover, we carefully design the construction such that
  $w \not\in M_\ell$.  Furthermore, the walk never moves far in the
  direction of $-\vec u$ because that would imply
  $w \not\in D_{\vec u, k}$.
  
  We begin by choosing $k_0\in\N$. We extend the vector $\vec u$ to an
  orthogonal basis $\vec b_1,\ldots,\vec b_n\in\Z^n$ of $\Q^n$,
  meaning that $\vec b_1=\vec u$ and
  $\langle \vec b_i,\vec b_j\rangle=0$ if $i\ne j$. Note that since
  $\vec b_i\ne 0$, we then have
  $\langle \vec b_i,\vec b_i\rangle=\|\vec b_i\|^2\ne 0$.  In
  particular, this means for every
  $\vec v\in P=\{\vec b_1,\vec b_2,-\vec b_2,\ldots, \vec b_n,-\vec
  b_n\}$, we have $\langle \vec v,\vec u\rangle\ge 0$. Note that
  except for $\vec b_1$, the set $P$ contains every vector $\vec b_i$
  positively and negatively.
  
  For each $i\in[1,2n-1]$, we pick a word $v_i\in\Sigma_n^*$ so that
  $\{\varphi(v_1),\ldots,\varphi(v_{2n-1})\}=P$.  Note that then, we
  have $\langle \varphi(v_i),\vec u\rangle\ge 0$ for every
  $i\in[1,2n-1]$.  Choose $k_0\in\N$ so that $k_0\ge 2|v_i|$ for each
  $i\in[1,2n-1]$.

  To show $D_{\vec u,k}\not\subseteq M_\ell\cup D_{\vec u_1,\ell_1}\cup\cdots\cup D_{\vec u_m,\ell_m}$, suppose $k\ge k_0$.
  We shall construct a word $w\in D_{\vec u,k}$ so that
  $w\notin M_\ell$ and $w\notin D_{\vec u_i,\ell_i}$ for every
  $i\in[1,m]$.  Pick $s\in\N$ with $s>\ell_i$ for $i\in[1,m]$. Write
  $\vec u=(x_1,\ldots,x_n)$ and let $u=a_1^{x_1}\cdots a_n^{x_n}$.
  Here, in slight abuse of notation, if $x_i<0$, we mean
  $\bar{a}_i^{|x_i|}$ instead of $a_i^{x_i}$.  Then clearly, we have
  $\varphi(u)=\vec u$ and every infix $z$ of $u$ satisfies
  $\langle \varphi(z),\vec u\rangle>0$.
  Let
  \[w=u^\ell v_1^{\ell\cdot s}v_2^{\ell\cdot s}\cdots v_{2n-1}^{\ell\cdot
      s}. \] Let us first show that $w\in D_{\vec u,k}$. Since
  $\langle \varphi(v_i),\vec u\rangle=\langle\vec b_i,\vec
  u\rangle=0$, we have
  $\langle \varphi(w),\vec u\rangle=\ell\cdot \langle \varphi(u),\vec
  u\rangle=\ell\cdot \|u\|^2>0$. Let $z$ be an infix of $w$.  Since
  $\langle \varphi(y),\vec u\rangle>0$ for every infix $y$ of $u$ and
  also $\langle \varphi(v_i),\vec u\rangle=0$ for $i\in[1,2n-1]$, we have
  $\langle \varphi(z),\vec u\rangle\ge -k_0\ge -k$. Thus, we have $w\in D_{\vec u,k}$.

  Finally, we prove that $w\notin M_\ell$ and
  $w\notin D_{\vec u_i,\ell_i}$ for $i\in[1,m]$.  First, note that
  $\varphi(w)\equiv \textbf{0}\bmod{\ell}$, so that $w\notin M_\ell$.  Let us
  now show that for $i\in[1,m]$, we have
  $w\notin D_{\vec u_i,\ell_i}$.  Since $\vec b_1,\ldots,\vec b_n$ is
  a basis of $\Q^n$, we can write
  $\vec u_i=\alpha_1\vec b_1+\cdots +\alpha_n\vec b_n$ for some
  $\alpha_1,\ldots,\alpha_n\in \Q$. Since the basis $\vec b_1,\ldots,\vec b_n$ is an orthogonal basis,
  we have
  \begin{align*}
    \langle \vec u_i, \vec b_j\rangle &=\langle \alpha_1\vec b_1+\cdots+\alpha_n\vec b_n,\vec b_j\rangle \\
                                      &=\alpha_1\langle \vec b_1,\vec b_j\rangle+\cdots+\alpha_n\langle\vec b_n,\vec b_j\rangle \\
                                      &=\alpha_j\langle \vec b_j,\vec b_j\rangle \\
                                      &=\alpha_j\cdot \|\vec b_j\|^2
  \end{align*}
  and thus $\langle \vec u_i,\vec b_j\rangle>0$ if and only if $\alpha_j>0$.
  
  Observe that now either $\alpha_1<0$ or $\alpha_j\ne 0$ for some
  $j\in[2,n]$: Otherwise, we would have
  $\vec u_i=\alpha_1\vec b_1=\alpha_1\vec u$ and thus
  $\Q_+\vec u_i=\Q_+\vec u$. Therefore, there is a vector
  $\vec p\in P$ with $\langle \vec u_i,\vec p\rangle < 0$.  Let
  $\vec p=\varphi(v_p)$ with $p\in[1,2n-1]$. Hence, the infix
  $v_p^{\ell\cdot s}$ of $w$ satisfies
  $\langle\varphi(v_p^{\ell\cdot s}),\vec u_i\rangle<-\ell s<-\ell_i$
  and hence $w\notin D_{\vec u_i,\ell_i}$.
  \qed
\end{proof}

\paragraph{Outline of the algorithm}
We now turn to the proof of \cref{result-vass-zvass}. According to
\cref{geometric-separability}, we have to decide whether a given VASS
language $L\subseteq\Sigma^*_n$ satisfies
$L\subseteq M_k\cup D_{\vec u_1,\ell_1}\cup\cdots\cup D_{\vec
  u_m,\ell_m}$ for some $k\in\N$, $\vec u_1,\ldots,\vec u_m\in\Z^n$
and $\ell_1,\ldots,\ell_m\in\N$. Our algorithm employs the KLMST
decomposition (so named by Leroux
and Schmitz~\cite{DBLP:conf/lics/LerouxS15} after its inventors) used by Sacerdote and
Tenney~\cite{sacerdote1977decidability},
Mayr~\cite{DBLP:conf/stoc/Mayr81},
Kosaraju~\cite{DBLP:conf/stoc/Kosaraju82}, and
Lambert~\cite{DBLP:journals/tcs/Lambert92} and recently cast in terms
of ideal decompositions by Leroux and
Schmitz~\cite{DBLP:conf/lics/LerouxS15}. The decomposition yields VASS
languages $L_1,\ldots,L_p$ with $L=L_1\cup\cdots\cup L_p$, together
with finite automata $\cA_1,\ldots,\cA_p$ whose languages
overapproximate $L_1,\ldots,L_p$, respectively. We show that the
$\cA_i$ are not only overapproximations, but are what we call
``modular envelopes'' (\cref{construct-modular-envelopes}). This
allows us to proceed similarly to the proof of
\cref{geometric-disjointness}. It suffices to check regular
separability for each $L_i$. If $\cone(\cA_i)=\Q^n$, then it suffices
to check whether $\sep{\Perms(L_i)}{Z_n}$ using
\cref{separability-commutative} (see \cref{lem:all-directions}).  If
$\cone(\cA_i)$ is contained in some halfspace, then we transform $L_i$
into a VASS language $\hat{L}_i\subseteq\Sigma_m^*$.  Here,
$\hat{L}_i$ essentially captures the walks of $L_i$ that stay close to
a strict linear subspace $U\subseteq\Q^n$ with $m=\dim U<n$.  Since
$m<n$, we can then apply our algorithm recursively to $\hat{L}_i$.

We first explain the concept of modular envelopes. We then describe
the algorithm for regular separability and finally, we show how to
construct modular envelopes.

\paragraph{Modular envelopes}

For a finite automaton $\cA$ with input alphabet $\Sigma$, let
$\Loop(\cA)\subseteq\Sigma^*$ be the set of words that can be read on
a cycle in $\cA$.  Recall that $\Parikh(w)$ denotes the Parikh image
of $w$. We say that an automaton $\cA$ is a \emph{modular envelope}
for a language $L\subseteq\Sigma^*$ if (i)~$L\subseteq L(\cA)$ and
(ii)~for every selection $u_1,\ldots,u_m\in\Sigma^*$ of words from
$\Loop(\cA)$ and every $w\in L$ and every $k\in\N$, there is a word
$w'\in L$ so that each $u_j$ is an infix of $w'$ and
$\Parikh(w')\equiv\Parikh(w)\bmod{k}$, where the congruence is defined
component-wise.

In other words, $\cA$ describes a regular overapproximation that is
  small enough that we can find every selection of $\cA$'s loops as infixes in a word
  from $L$ 
  whose Parikh image is congruent modulo $k$ to a given word from
  $L$. Using the KLMST decomposition, we prove:
\begin{restatable}{theorem}{constructModularEnvelopes}\label{construct-modular-envelopes}
  Given a VASS language $L$, one can construct VASS languages
  $L_1,\ldots,L_p$, together with a modular envelope $\cA_i$ for each
  $L_i$ such that $L=L_1\cup\cdots\cup L_p$.
\end{restatable}
We postpone the proof of \cref{construct-modular-envelopes} until
\cref{envelopes} and first show how
it is used to decide geometric separability.

\paragraph{Modular envelopes with cone $\Q^n$} By
\cref{lem:dichotomy} we know that every cone either equals $\Q^n$ or
is included in some half-space. The following lemma will be useful in
the first case.
\begin{lemma}\label{lem:all-directions}
  Let $L\subseteq\Sigma_n^*$ be a language with a modular envelope $\cA$.  If
  $\cone(\cA) = \Q^n$ then the following are equivalent:
  (i)~$\sep{L}{Z_n}$,
  (ii)~$L \subseteq M_k$ for some $k\in\N$,
  (iii)~$\sep{\Perms(L)}{Z_n}$.
\end{lemma}
\begin{proof}
  Note that (ii) implies (iii) immediately and that (iii) implies (i)
  because $L\subseteq\Perms(L)$. Thus, we only need to show that (i)
  implies (ii).  By \cref{geometric-separability} if $\sep{L}{Z_n}$,
  then $L$ is included in some
  $M_k \cup D_{\vec u_1,k} \cup \cdots \cup D_{\vec u_m,k}$.  We show
  that in our case we even have $L \subseteq M_k$.

  Take any $w \in L$. We aim at constructing $w' \in L$ such that
  $w' \notin D_{\vec u_i,k}$ for every $i \in [1,m]$ and
  additionally $\varphi(w') \equiv \varphi(w) \bmod{k}$.  Since
  $\cone(\cA) = \Q^n$, for every $\vec u_j$, there exist a loop $v_j$
  in $\cA$ such that $\scal{\varphi(v_j)}{\vec u_j} < 0$.  Since
  $\varphi(v_j),\vec u_j\in\Z^n$, we even have
  $\scal{\varphi(v_j)}{\vec u_j} \leq -1$.  Since each $v_j$ belongs
  to $\Loop(\cA)$, the words $v_j^{k+1}$ also belong to $\Loop(\cA)$.

  Since $\cA$ is a modular envelope, there exists a word $w' \in L$
  that has all the words
  $v_1^{k+1}, \ldots, v_{m}^{k+1}$ as infixes and where
  $\Parikh(w) \equiv \Parikh(w') \bmod{k}$.
  Recall that every infix $u$ of every word in $D_{\vec u_i, k}$ has
  $\scal{\varphi(u)}{\vec u_j} \geq -k$. However, we have the inequality
  $\scal{\varphi(v_j^{k+1})}{\vec u_j} \leq -(k+1)$. Thus, $w'$ cannot
  belong to $D_{\vec u_i,k}$ for $i \in [1,m]$.  Since
  $L\subseteq M_k\cup D_{\vec u_1,k}\cup\cdots \cup D_{\vec u_m,k}$,
  this only leaves $w' \in M_k$. Since
  $\Parikh(w')\equiv \Parikh(w)\bmod{k}$, we also have
  $\varphi(w')\equiv\varphi(w)\bmod{k}$ and thus $w\in M_k$.
\end{proof}

\begin{algorithm}
  \DontPrintSemicolon
  \SetKwInOut{Input}{Input}
  \SetKw{Exit}{exit}
  \Input{$n\in\N$ and VASS language $L=\Lang{V}\subseteq\Sigma_n^*$}
  \caption{Deciding separability of a VASS language $L$ from $Z_n$}\label{algorithm}

  \lIf{$n=0$ and $L=\emptyset$}{\Return{``yes''}}
  \lIf{$n=0$ and $L\neq \emptyset$}{\Return{``no''}}
  Use KLMST decomposition to compute VASS languages $L_1,\ldots,L_p$,
  together with modular envelopes $\cA_1,\ldots,\cA_p$.\\
  \For{$i\in[1,p]$}{
    \If{$\cone(\cA_i)=\Q^n$}{
      Check whether $\Perms(L_i)|Z_n$\\
      \lIf{not $\Perms(L_i)|Z_n$}{
        \Return{``no''}
      }
    }
    \If{$\cone(\cA_i)\subseteq H=\{\vec x\in\Q^n\mid \langle \vec x,\vec u\rangle\ge 0\}$ for some $\vec u\in\Z^n\setminus \{\bf{0}\}$}{
      Let $U=\{\vec x\in \Q^n\mid \langle \vec x,\vec u\rangle=0\}$. \tcc*[f]{$\dim U=n-1$} \\
      Compute $k,\ell$ with $L_i \subseteq\Lang{\cA_i}\subseteq D_{\vec u,k}\cup S_{U,\ell}$
      \tcc*[f]{Now we have $L_i\setminus D_{\vec u,k}\subseteq S_{U,\ell}$}\\
      Compute transduction $T_{U,\ell}\subseteq \Sigma_n^*\times\Sigma_{n-1}^*$ \\
      Compute VASS for $\hat{L}_i=T_{U,\ell}(L_i\!\setminus\! D_{\vec u,k})$ \\
      \tcc*[f]{$\hat{L}_i$ has dimension $n-1$: $\hat{L}_i\subseteq\Sigma_{n-1}^*$}\\
      Check recursively whether $\hat{L}_i|Z_{n-1}$  \\
      \lIf{not $\hat{L}_i|Z_{n-1}$}{
        \Return{``no''}
      }
    }
  }
  \Return{``yes''}
  
\end{algorithm}

We are now prepared to explain the decision procedure for
\cref{result-vass-zvass}. 
The algorithm is illustrated in
\cref{algorithm}.  If $n=0$, then $\Sigma_n=\emptyset$ and thus either
$L=\emptyset$ or $L=\{\varepsilon\}$, meaning $\sep{L}{Z_0}$ if and
only if $L\ne\emptyset$. If $n\ge 1$, we perform the KLMST
decomposition, which, as explained in \cref{envelopes}, yields
languages $L_1\cup\cdots\cup L_p$ and modular envelopes
$\cA_1,\ldots,\cA_p$ such that $L=L_1\cup\cdots\cup L_p$. Since then
$\sep{L}{Z_n}$ if and and only if $\sep{L_i}{Z_n}$ for each
$i\in[1,p]$, we check for the latter. For each $i\in[1,p]$, the
dichotomy of \cref{lem:dichotomy} guides a case distinction: If
$\cone(\cA_i)=\Q^n$, then by \cref{lem:all-directions},
$\sep{L_i}{Z_n}$ if and only if $\sep{\Perms(L_i)}{Z_n}$, which can be
checked via \cref{separability-commutative}.

If $\cone(\cA_i)$ is contained in some half-space
$H=\{\vec x\in\Q^n \mid \langle x,\vec u\rangle\ge 0\}$ with
$\vec u\in\Z^n\setminus\{\textbf{0}\}$, then \cref{halfspace} tells us that
$L_i\subseteq \Lang{\cA_i}\subseteq D_{\vec u,k}\cup S_{U,\ell}$ for
some computable $k,\ell\in\N$ and
$U=\{\vec x\in\Q^n \mid \langle \vec x,\vec u\rangle=0\}$.  In
particular, we have $\sep{L_i}{Z_n}$ if and only if
$\sep{L_i\!\setminus\! D_{\vec u,k}}{Z_n}$.  Note that
$L_i\setminus D_{\vec u,k}=L_i\cap (\Sigma^*_n\setminus D_{\vec u,k})$
is a VASS language and is included in $S_{U,\ell}$.  Thus, the walks
in $L_i$ always stay close to the hyperplane $U$, which has dimension
$n-1$. We can therefore use the transduction $T_{U,\ell}$ to transform
$L_i$ into a set $\hat{L}_i$ of walks in $(n-1)$-dimensional space and
decide separability recursively for the result: We have
$\hat{L}_i=T_{U,\ell}(L_i\setminus D_{\vec
  u,k})\subseteq\Sigma_{n-1}^*$ and \cref{dimension-reduction} tells
us that $\sep{L_i}{Z_n}$ if and only if $\sep{\hat{L}_i}{Z_{n-1}}$.

\subsection{Constructing modular envelopes}\label{envelopes}

\paragraph{Petri nets} We now prove \cref{construct-modular-envelopes}. Since our proof crucially relies on Lambert's
iteration lemma (\cref{lambert-pumping}), we adopt in this section the
notation of Lambert and phrase our proof in terms of Petri nets.  A
\emph{Petri net} $N = (P, T, \vec \Pre, \vec \Post)$ consists of a
finite set $P$ of \emph{places}, a finite set $T$ of
\emph{transitions} and two mappings
$\vec \Pre, \vec \Post\colon T \to \N^P$.  Configurations of Petri net
are elements of $\N^P$, called \emph{markings}.
The \emph{effect} of a transition $t \in T$ is $\vec \Post(t) - \vec \Pre(t) \in \Z^P$, denoted $\eff(t)$.
If for every place $p \in P$ we have $\vec \Pre(t)[p] \leq \vec M[p]$ for a transition $t \in T$ then
$t$ is \emph{fireable} in $\vec M$ and the result of firing $t$ in marking $\vec M$ is $\vec M' = \vec M + \eff(t)$,
we write $\vec M \trans{t} \vec M'$. We extend notions of fireability and firing naturally to sequences
of transitions, we also write $\vec M \trans{w} \vec M'$ for $w \in T^*$.
The \emph{effect} of $w \in T^*$, $w=t_1\cdots t_m$, $t_1,\ldots,t_m\in T$ is $\eff(w) = \eff(t_1)+\cdots+\eff(t_m)$.

For a Petri net $N = (P, T, \vec \Pre, \vec \Post)$ and markings
$\vec M_0, \vec M_1$, we define the language
$L(N, \vec M_0, \vec M_1) = \{w \in T^* \mid \vec M_0 \trans{w} \vec
M_1\}$.  Hence, $L(N, \vec M_0, \vec M_1)$ is the set of transition
sequences leading from $\vec M_0$ to $\vec M_1$.  A \emph{labeled
  Petri net} is a Petri net $N=(P,T,\vec \Pre,\vec \Post)$ together
with an \emph{initial marking} $\vec M_I$, a \emph{final marking}
$\vec M_F$, and a \emph{labeling}, i.e.  a homomorphism
$h: T^*\to\Sigma^*$.  The language \emph{recognized by} the labeled
Petri net is then defined as
$L_h(N, \vec M_I, \vec M_F) = h(L(N, \vec M_I, \vec M_F))$.

It is folklore (and easy to see) that a language is a VASS language if
and only if it is recognized by a labeled Petri net (and the
translation is effective). Thus, it suffices to show
\cref{construct-modular-envelopes} for languages of the form
$L=h(L(N,\vec M_I,\vec M_F))$.  Moreover, it is already enough to
prove \cref{construct-modular-envelopes} for languages of the form
$L(N,\vec M_I,\vec M_F)$: If $\cA$ is a modular envelope for $L$, then
applying $h$ to the edges of $\cA$ yields a modular envelope for
$h(L)$.  Thus from now on, we assume $L=L(N,\vec M_I,\vec M_F)$ for a
fixed Petri net $N=(P,T,\vec \Pre,\vec \Post)$.

\paragraph{Basic notions} Let us introduce some notions used in Lambert's proof.
We extend the set of configurations $\N^d$ into $\Nom^d$, where $\Nom=\N\cup\{\omega\}$.
We extend the notion of transition firing into $\Nom^d$, by
defining $\omega - k = \omega = \omega + k$ for every $k \in \N$. For $\vec u, \vec v \in \Nom^d$
we write $\vec u \leq_\omega \vec v$ if
$\vec u[i] = \vec v[i]$ or $\vec v[i] = \omega$. Intuitively reaching a configuration with $\omega$ at some places
means that it is possible to reach configurations with values $\omega$ substituted by arbitrarily high values.

A key notion in~\cite{DBLP:journals/tcs/Lambert92} is that of MGTS, which formulate restrictions on paths in
Petri nets. A \emph{marked graph-transition sequence (MGTS)} for our Petri net
$N=(P,T,\vec \Pre,\vec \Post)$ is a finite sequence
$
C_0, t_1, C_1 \ldots C_{n-1}, t_n, C_n,
$
where $t_i$ are transitions from $T$ and $C_i$ are precovering graphs, which are defined next.
A \emph{precovering graph} is a quadruple $C = (G, \vec m, \vec m^\init, \vec m^\fin)$, where $G=(V,E,h)$ is a finite, strongly connected, directed graph with $V\subseteq\Nom^P$ and labeling $h\colon E \to T$,
and three vectors: a \emph{distinguished} vector $\vec m \in V$, an \emph{initial} vector $\vec m^\init \in \Nom^P$, and a \emph{final} vector $\vec m^\fin \in \Nom^P$.
A precovering graph has to meet two conditions:
First, for every edge $e = (\vec m_1, \vec m_2) \in E$, there is an $\vec m_3\in\Nom^P$
with $\vec m_1 \trans{h(e)} \vec m_3 \leq_\omega \vec m_2$.
Second, we have
$\vec m^\init, \vec m^\fin \leq_\omega \vec m$.
Additionally we impose the restriction on MGTS that the initial vector of $C_0$
equals $\vec M_I$ and the final vector of $C_n$ equals $\vec M_F$.

\paragraph{Languages of MGTS} Each precovering graph can be
treated as a finite automaton. For $\vec m_1, \vec m_2\in V$,
 $L(C, \vec m_1, \vec m_2)$ denotes the set of all $w\in T^*$ read on
a path from $\vec m_1$ to $\vec m_2$.  Moreover, let
$L(C) = L(C, \vec m, \vec m)$.  MGTS have associated languages as
well.  Let $\Ncal = C_0, t_1, C_1 \ldots C_{n-1}, t_n, C_n$ be an MGTS
of a Petri net $N$, where
$C_i = (G_i, \vec m_i, \vec m_i^\init, \vec m_i^\fin)$. Its language
$L(\Ncal)$ is the set of all words of the form
$w = w_0 t_1 w_1 \cdots w_{n-1} t_n w_n \in T^*$ where:
$w_i \in L(C_i)$ for each $i\in[0,n]$ and (ii)~there exist markings
$\vec u_0, \vec u'_0, \vec u_1, \vec u'_1, \ldots, \vec u_n, \vec u'_n
\in \N^P$ such that $\vec u_i \leq_\omega \vec m^\init_i$ and
$\vec u'_i \leq_\omega \vec m^\fin_i$ and
\begin{equation} \vec u_0 \trans{w_0} \vec u'_0 \trans{t_1} \vec u_1 \trans{w_1} \ldots \trans{w_{n-1}} \vec u'_{n-1} \trans{t_n} \vec u_n \trans{w_n} \vec u'_n. \label{mgts-run} \end{equation}
In this situation, the occurrences of $t_1,\ldots,t_n$ shown in \cref{mgts-run}
are called the \emph{bridges} and $w_0,\ldots,w_n$ are called the \emph{graph parts}.
 Notice that by (ii) and the restriction that
 $\vec m^\init_0 = \vec M_I$ and $\vec m^\fin_n = \vec M_F$, we have
 $L(\Ncal)\subseteq L(N, \vec M_I,\vec M_F)$ for any MGTS $\Ncal$.
Hence roughly speaking, $L(\Ncal)$ is the set of runs that contain the transitions
$t_1,\ldots,t_n$ and additionally markings before and after firing these
transitions are prescribed on some places: this is exactly what the restrictions
$\vec u_i \leq_\omega \vec m^\init_i$, $\vec u'_i \leq_\omega \vec m^\fin_i$ impose.
%
As an immediate consequence of the definition, we observe that for every
MGTS $\Ncal = C_0, t_1, C_1 \ldots C_{n-1}, t_n, C_n$
we have 
\begin{equation} L(\Ncal) \subseteq L(C_0) \cdot \{t_1\} \cdot L(C_1) \cdots L(C_{n-1}) \cdot \{t_n\} \cdot L(C_n). \label{eq:overapproximation}\end{equation}

\paragraph{Perfect MGTS}
Lambert calls MGTS with a particular property
\emph{perfect}~\cite{DBLP:journals/tcs/Lambert92}. Since the precise
definition is involved and we do not need all the details, it is
enough for us to mention a selection of facts about perfect MGTS
(\cref{perfect-covering-sequence,thm:language-union}). One of these is
the intuitive property that in perfect MGTSes, the value $\omega$ on
place $p$ in $\vec m_i$ means that inside of the precovering graph
$C_i$, the token count in place $p$ can be made arbitrarily high.  The
precise formulation involves the notion of covering sequences, which
we define next.  Let $C$ be a precovering graph for a Petri net
$N = (P, T, \vec \Pre, \vec \Post)$ with a distinguished vector
$\vec m \in \Nom^P$ and initial vector $\vec m^{\init} \in \Nom^P$.
For a marking $\vec M_0$ let
$L(N, \vec M_0) = \bigcup_{M \in \N^P} L(N, \vec M_0, \vec M)$,
i.e. the set of all the transition sequences fireable in $\vec M_0$.
A sequence $x \in L(C) \cap L(N, \vec m^{\init})$ is called a
\emph{covering sequence for $C$} if $x$ is enabled in $\vec m^{\init}$
and for every place $p \in P$ we have either (i)~$\vec m^{\init}[p] = \omega$, or (ii)~$\vec m[p] = \vec m^{\init}[p]\in\N$ and $\eff(x)[p] = 0$, or (iii)~$\vec m^{\init}[p]<\vec m[p] = \omega$ and $\eff(x)[p] > 0$.  The
property of perfect MGTS that we need is the following:
\begin{property}\label{perfect-covering-sequence}
  In a perfect MGTS $\Ncal$, each precovering graph possesses a
  covering sequence.
\end{property}
This is part of the definition of perfect MGTS,
see~\cite[page~92]{DBLP:journals/tcs/Lambert92}.
The second fact about perfect MGTS that we will use is
that one can decompose each Petri net into finitely many perfect MGTS.
In~\cite{DBLP:journals/tcs/Lambert92} the following is shown (Theorem~4.2
(page~94) together with the preceding definition).
\begin{proposition}[\cite{DBLP:journals/tcs/Lambert92}]\label{thm:language-union}
  Given a Petri net $N$, one can compute finitely many perfect MGTS
  $\Ncal_1, \ldots, \Ncal_p$ such that $L(N,\vec M_I,\vec M_F)$ equals
  $\bigcup_{i=1}^p L(\Ncal_i)$.
\end{proposition}

\paragraph{Building the automata} By
\cref{thm:language-union}, it suffices to construct a modular envelope
for each $L(\Ncal_i)$. Hence, we consider a single perfect MGTS
$\Ncal=C_0,t_1,C_1,\ldots, t_n,C_n$ with distinguished vertices
$\vec m_0,\ldots,\vec m_n$ and construct a modular envelope $\cA$ for
$L(\Ncal)$.  We obtain $\cA$ by gluing together all
the precovering graphs $C_i$ along the transitions $t_i$. In other
words, $\cA$ is the disjoint union of all the graphs $C_i$ and has an
edge labeled $t_i$ from $\vec m_{i-1}$ to $\vec m_i$ for each
$i\in[1,n]$. The initial state of $\cA$ is $\vec m_0$ and its final
state is $\vec m_n$.

\paragraph{Ingredient I: Run amalgamation}
The first ingredient in for showing that $\cA$ is a modular envelope
is a method for constructing runs in Petri nets: the amalgamation of
runs as introduced by Leroux and
Schmitz~\cite{DBLP:conf/lics/LerouxS15}. It is based on an embedding
between Petri net runs introduced by
Jan\v{c}ar~\cite{DBLP:journals/tcs/Jancar90} and
Leroux~\cite{DBLP:conf/popl/Leroux11}.  A triple
$(\vec u,t,\vec v)\in \N^P\times T\times \N^P$ is a \emph{transition
  triple} if $\vec v=\vec u+\eff(t)$. If there is no danger of
confusion, we sometimes call $(\vec u,t,\vec v)$ a transition.  A
triple $(\vec u,w,\vec v)$ with $\vec u,\vec v\in\N^P$ and
$w\in (\N^P\times T\times \N^P)^*$ is called a \emph{prerun}.
Let $\rho=(\vec u,w,\vec v)$ and $\rho'=(\vec u',w',\vec v')$ be
preruns with
$w=(\vec u_0,t_1,\vec v_1)(\vec u_1,t_2,\vec v_2)\cdots (\vec
u_{r-1},t_r,\vec v_r)$ and
$w'=(\vec u'_0,t'_1,\vec v'_1)(\vec u'_1,t'_2,\vec v'_2)\cdots (\vec
u'_{s-1},t'_s,\vec v'_s)$.  An \emph{embedding of $\rho$ in $\rho'$}
is a monotone map $\sigma\colon[1,r]\to[1,s]$ such that $t'_{\sigma(i)}=t_i$,
$\vec u_i\le\vec u'_{\sigma(i)}$ and $\vec v_i\le\vec v'_{\sigma(i)}$
for $i\in[1,r]$, and $\vec u\le\vec u'$ and $\vec v\le\vec v'$.  In
this case, we call the words
$t'_1\cdots t'_{\sigma(1)-1}$,
$t'_{\sigma(i)+1}\cdots t'_{\sigma(i+1)-1}$ for $i\in[1,r-1]$, and
$t'_{\sigma(r)+1}\cdots t'_s$
the \emph{inserted words of $\sigma$}. By $F(\sigma)\subseteq T^*$, we
denote the set of all infixes of inserted words of $\sigma$.
Furthermore, by $\Parikh(\rho)$, we denote the Parikh image
$\Parikh(t_1\cdots t_r)\in\N^T$.

Moreover, $\rho$ is called a \emph{run} if each
$(\vec u_i,t_i,\vec v_i)$ is a transition and also $\vec u=\vec u_0$,
$\vec u_i=\vec v_i$ for $i\in[1,r]$, and $\vec v=\vec v_r$.  Note that
this is equivalent to $\vec u_i=\vec v_i$ for $i\in[1,r]$ and
$\vec u=\vec u_0\trans{t_1}\vec u_1\cdots \vec u_{r-1}\trans{t_r}\vec
u_r$ and we sometimes use the latter notation to denote runs.

Suppose we have three runs $\rho_0,\rho_1,\rho_2$ and there are
embeddings $\sigma_1$ of $\rho_0$ in $\rho_1$ and $\sigma_2$ of
$\rho_0$ in $\rho_2$. As observed in
\cite[Prop.~5.1]{DBLP:conf/lics/LerouxS15} one can define a new run
$\rho_3$ in which both $\rho_1$ and $\rho_2$ embed. Let $\rho_0$ be the run
$ \vec u_0\trans{t_1}\vec u_1\trans{t_2}\cdots \trans{t_r}\vec u_r$.
Then $\rho_1$ and $\rho_2$ can be written as
\begin{equation}\label{runs-in-amalgamation}
\begin{aligned}
  \rho_1\colon \vec u_0+\vec v_0\trans{w_0}\vec u_0+\vec v_1 &\trans{t_1}\vec u_1+\vec v_1\trans{w_1}\vec u_1+\vec v_2\cdots \\
  \cdots&\trans{t_r} \vec u_r+\vec v_r\trans{w_r}\vec u_r+\vec v_{r+1} \\
  \rho_2\colon \vec u_0+\vec v'_0\trans{w'_0}\vec u_0+\vec v'_1 &\trans{t_1}\vec u_1+\vec v'_1\trans{w'_1}\vec u_1+\vec v'_2\cdots \\
  \cdots&\trans{t_r} \vec u_r+\vec v'_r\trans{w'_r}\vec u_r+\vec v'_{r+1}
\end{aligned}
\end{equation}
for some $\vec v_i,\vec v'_i\in\N^P$, $i\in[0,r+1]$. Then the
\emph{amalgam of $\rho_1$ and $\rho_2$ (along $\sigma_1$ and
  $\sigma_2$)} is the run $\rho_3$ defined as
\begin{equation}
\begin{aligned}
  & \vec u_0+\vec v_0+\vec v'_0 \trans{w_0} \vec u_0+\vec v_1+\vec v'_0 \trans{w'_0} \vec u_0+\vec v_1+\vec v'_1 \trans{t_1} \\
 & \vec u_1+\vec v_1+\vec v'_1 \trans{w_1} \vec u_1+\vec v_2+\vec v'_1 \trans{w'_1} \vec u_1+\vec v_2+\vec v'_2 \trans{t_1} \\
  &\vdots \\
  & \vec u_{r-1}+\vec v_{r-1}+\vec v'_{r-1} \trans{w_1} \vec u_{r-1}+\vec v_r+\vec v'_{r-1} \trans{w'_1} \vec u_{r-1}+\vec v_r+\vec v'_r \trans{t_r} \\
  &\vec u_r+\vec v_r+\vec v'_r\trans{w_r}\vec u_r+\vec v_{r+1}+\vec v'_{r}\trans{w'_r} \vec u_r+\vec v_{r+1}+\vec v'_{r+1}.
\end{aligned}\label{amalgamated-run}
\end{equation}
and the embedding $\tau$ of $\rho_0$ in $\rho_3$ is defined in the
obvious way.  Note that the run $\rho_3$ and the embedding $\tau$
satisfy
\begin{equation}\label{amalgamation-properties}
  \begin{aligned}
  F(\sigma_1)\cup F(\sigma_2)&\subseteq F(\tau),\\
  \Parikh(\rho_3)-\Parikh(\rho_0)&=(\Parikh(\rho_1)-\Parikh(\rho_0))+(\Parikh(\rho_2)-\Parikh(\rho_0)).
  \end{aligned}
\end{equation}

The following \lcnamecref{amalgam-in-mgts} is very much in the spirit
of Leroux and Schmitz~\cite{DBLP:conf/lics/LerouxS15}, which recasts
the KLMST algorithm as the computation of an ideal
decomposition. Specifically, their
\cite[Lemma~VII.2]{DBLP:conf/lics/LerouxS15} shows that the set of
runs of $\Ncal$ is upward directed, meaning that for any two runs
$\rho_1$ and $\rho_2$, there exists $\rho_3$ in $\Ncal$ in which both
$\rho_1$ and $\rho_2$ embed. We need precise control over the Parikh
image of the runs we construct. Therefore, we introduce the notion of
compatible embeddings, which guarantees that the amalgam of two
runs from $\Ncal$ again belongs to $\Ncal$.

If $\rho$ and $\rho'$ are runs in $\Ncal$, then we can associate to
each marking $\vec u_i$ ($\vec u'_i$) in $\rho$ (in $\rho'$) a node
$\tilde{\vec v}_i$ ($\tilde{\vec v}'_i$) in some $C_j$.  We say that
$\sigma$ is \emph{($\Ncal$-)compatible} if (i)~$\sigma$ maps the
$k$-th bridge transition in $\rho$ to the $k$-th bridge transition in
$\rho'$ for each $k\in[1,n]$ and (ii)~for $i\in[1,r]$, we have
$\tilde{\vec v}'_{\sigma(i)}=\tilde{\vec v}_i$.  In other words,
$\sigma$ does (i)~preserve bridge transitions and (ii)~map each
marking in $\rho$ to a marking in $\rho'$ that visits the same node in
$\Ncal$.
\begin{restatable}{lemma}{amalgamInMGTS}\label{amalgam-in-mgts}
  Let $\rho_0, \rho_1, \rho_2$ be runs in $\Ncal$ where $\rho_0$
  embeds compatibly in $\rho_1$ and $\rho_2$. Then the amalgam
  $\rho_3$ of $\rho_1$ and $\rho_3$ is also a run in
  $\Ncal$. Moreover, the induced embeddings of $\rho_1,\rho_2$ in
  $\rho_3$ are compatible.
\end{restatable}

For \cref{amalgam-in-mgts}, we roughly argue as follows. For runs as
in \cref{runs-in-amalgamation}, compatibility means that $\vec u_i$
and $\vec u_i+\vec v_{i+1}$ and $\vec u_i+\vec v'_{i+1}$ correspond to
the same node in a component of $\Ncal$. Therefore, the difference
vectors $\vec v_{i+1}$ and $\vec v'_{i+1}$ can be non-zero only in
coordinates that have $\omega$ in these nodes. This means, the vector
$\vec u_i+\vec v_{i+1}+\vec v'_{i+1}$ also differs from $\vec u_i$ in
only those coordinates and can again be associated to the same node to
show that $\rho_3$ is a run in $\Ncal$.

\begin{proof}
  Let $\rho_0$ be the run
  $\vec u_0\trans{t_1}\vec u_1\trans{t_2}\cdots \trans{t_r}\vec u_r$
  and let $\rho_1$ and $\rho_2$ be as in \cref{runs-in-amalgamation}.
  Moreover, 
  let $\tau$ be the resulting embedding of $\rho_0$ in $\rho_3$ and let $\sigma'_j$ be
  the embedding of $\rho_j$ in $\rho_3$ for $j\in\{1,2\}$.
  Let us first argue that $\rho_3$ is a run in $\Ncal$.  To this end,
  we argue that the images of bridge transitions under $\tau$ satisfy
  condition~(i) of a run. Suppose $t_i$ in $\rho_0$ is a bridge
  transition. Then for some final marking $\vec m^{\fin}$ of some
  graph $C_j$ and some initial marking $\vec m^{\init}$ of $C_{j+1}$,
  we have $\vec u_{i-1}\le_\omega \vec m^{\fin}$ and
  $\vec u_{i}\le_\omega \vec m^{\init}$. Note that $\vec m^{\fin}$ and
  $\vec m^{\init}$ have $\omega$ in the same set of places; we denote
  this set by $\Omega\subseteq P$.  Since $\sigma_1$ is compatible and
  $\rho_1$ is a run, this also implies
  $\vec u_{i-1}+\vec v_{i}\le_\omega \vec m^{\fin}$ and
  $\vec u_{i}+\vec v_{i}\le_\omega \vec m^{\init}$. Note that this
  means $\vec v_i[p]=0$ if $p\in P\setminus \Omega$, in other words
  $\vec v_{i}\in\N^\Omega$. By the same argument, we have
  $\vec v'_i\in\N^\Omega$.  Therefore, we also have
  $\vec u_{i-1}+\vec v_i+\vec v'_i\le \vec m^{\fin}$ and
  $\vec u_i+\vec v_i+\vec v'_i\le \vec m^{\init}$, which proves
  condition~(i).

  We now continue with condition~(ii). Since $\sigma_1$ and $\sigma_2$
  are compatible, the bridge transitions of $\rho_0$ must be among the
  $t_1,\ldots,t_r$.  Therefore, each $w_i$ is included in some graph
  part of $\rho_1$; and each $w'_i$ is included in some graph
  part of $\rho_2$.  Let
  $\tilde{\vec u}_i\in\N_\omega^P$ be the node in the graph $C$
  associated with $\vec u_i$ in $\rho_0$. Since $\sigma_1$ is
  compatible, the $\tilde{\vec u}_i$ is also the node associated with
  $\vec u_i+\vec v_i$ and with $\vec u_i+\vec v_{i+1}$. In particular,
  we have $w_i\in L(C,\tilde{\vec u}_i)$.  By the same argument, we
  have $w'_i\in L(C,\tilde{\vec u}_i)$ and therefore
  $w_iw'_i\in L(C,\tilde{\vec u}_i)$. In other words, $\rho_3$ is
  obtained from $\rho_2$ by inserting loops $w'_i$ in graphs right
  after loops at the same node $\tilde{\vec u}_i$. Thus, $\rho_3$ is a
  run in $\Ncal$.

  It remains to be shown that $\sigma'_1,\sigma'_2$ are compatible. By
  symmetry, it suffices to show this for $\sigma'_1$.  Consider the
  transitions in $\rho_1$ inside
  $\vec u_i+\vec v_i\trans{w_i}\vec u_i+\vec v_{i+1}$.  Since the
  bridge transitions of $\rho_1$ must be among $t_1,\ldots,t_r$, we
  know that $w_i$ is included in some graph part of $\rho_1$ in some
  graph $C$.  Observe that because of strong connectedness, all nodes
  in a graph of an MGTS have $\omega$ in exactly the same places.  Let
  $\Omega\subseteq P$ be the set of those places.  Let
  $\tilde{\vec u}_i\in\N_\omega^P$ be the node associated to $\vec u_i$ in
  $\rho_0$. Since $\sigma_2$ is compatible, the marking
  $\vec u_i+\vec v'_i$ is associated with the same node $\tilde{\vec u}_i$.
  Therefore, we have $\vec u_i\le_\omega \tilde{\vec u}_i$ and
  $\vec u_i+\vec v_i\le_\omega \tilde{\vec u}_i$. This means for every
  $p\in P$ with $\vec v'_i[p]\ne 0$, we have $\tilde{\vec u}_i[p]=\omega$.
  In other words, $\vec v'_i\in\N^\Omega$.  Note that every node
  associated to a transition in $w_i$ belong to $C$ and adding a
  vector from $\N^\Omega$ to a marking does not change its associated
  node. Since $\sigma'_1$ maps
  $\vec u_i +\vec v_i\trans{w_i}\vec u_i+\vec v_{i+1}$ to
  $\vec u_i+\vec v_i+\vec v'_i \trans{w_i} \vec u_i+\vec v_{i+1}+\vec
  v'_i$, $\sigma'_1$ has to preserve the nodes of the transitions in
  $w_i$. Similarly, one shows that $\sigma'_1$ preserves nodes of
  markings around non-bridge transitions among $t_1,\ldots,t_r$.
\end{proof}

\paragraph{Ingredient II: Lambert's iteration lemma}
Our second ingredient for showing that $\cA$ is a modular envelope is
Lambert's iteration lemma. It allows us to construct runs containing
desired infixes, which can then be revised using amalgamation.  Recall
that we consider the marked graph-transition sequence
$\Ncal = C_0, t_1, C_1 \ldots C_{n-1}, t_n, C_n$.  Let
$C_i = (V_i, E_i, h_i)$ be a precovering graph, and let the
distinguished vertex be $\vec m_i$ and initial vertex be
$\vec m_i^\init$. The following is a simplified version of Lambert's
iteration lemma (Lemma 4.1 in~\cite{DBLP:journals/tcs/Lambert92}
(page~92)).
\begin{lemma}[Lambert~\cite{DBLP:journals/tcs/Lambert92}]\label{lambert-pumping}
  Suppose $\Ncal=C_0,t_1,C_1,\ldots,t_n,C_n$ and let $x_i\in T^*$ be a
  covering sequence for $C_i$ for $i\in[0,n]$.  Then there exist
  $k_0\in\N$ and sequences $\beta_i,y_i,z_i\in T^*$ for $i\in[0,n]$
  such that for every $k\ge k_0$,
  \[ x_0^k\beta_0y_0^k z_0^k \cdot t_1\cdot x_1^k\beta_1 y_1^k
    z_1^k\cdots t_n\cdot x_n^k \beta_n y_n^kz_n^k \] is a run in
  $\Ncal$, such that the shown occurrences of $t_1,\ldots,t_n$ are the
  bridges.  Moreover, we have $\sum_{j=0}^i \eff(x_iy_iz_i)[p]\ge 1$
  for every $i\in[1,n]$ and $p\in P$ with $\vec m_i^{\fin}[p]=\omega$.
\end{lemma}
The inequalities $\sum_{j=0}^i \eff(x_iy_iz_i)[p]\geq1$ for each $p$ with
$\vec m^{\fin}_i[p]=\omega$ follow from item~(i) in
\cite[Lemma~4.1]{DBLP:journals/tcs/Lambert92}: In the notation of \cite{DBLP:journals/tcs/Lambert92}, item~(i) of Lemma~4.1 states that
the effect of $x_0y_0z_0\cdots x_jy_jz_j$ in
$p$ equals $\alpha(x_0(c'_j(p))-x_0(c_0(p)))$. Since
$x_0$ is a solution to the homogeneous characteristic equation (see
page~91 for the definition), we have
$x_0(c_0(p))=0$.  Perfectness (see page~92 for the definition) and the
fact that
$x_0$ even has maximal support among all solutions imply that
$x_0(c'_j(p))\ge 1$ for such
$p$. Finally, the proof of
\cite[Lemma~4.1]{DBLP:journals/tcs/Lambert92} chooses
$\alpha\in\N$ so as to be above certain thresholds. We may therefore
assume that $\alpha\ge 1$.

We use \cref{lambert-pumping} to construct a run in $\Ncal$
that contains the desired infixes and in which a given run embeds.
\begin{restatable}{lemma}{injectFactors}\label{inject-factors}
  For every run $\rho$ in $\Ncal$ and
  words $u_1,\ldots,u_m\in\Loop(\cA)$, there is a run $\rho'$ in
  $\Ncal$ such that $\rho$ embeds in $\rho'$ via a compatible
  embedding $\sigma$ with $u_1,\ldots,u_m\in F(\sigma)$.
\end{restatable}

\begin{proof}
Roughly speaking, the proof of \cref{inject-factors} proceeds as
follows.  First, we take a covering sequence for every component and
show that they can be prolongated so that each $u_i$ appears as an
infix of some covering sequence. Then, we prolongate the covering
sequences further so that they contain the $w_i$ if $\rho$ is as in
\cref{mgts-run}. Next, we iterate each covering sequence so that it
creates enough tokens in places $p$ with $\vec m_i[p]=\omega$,
$\vec m_{i-1}^{\fin}[p]\in\N$ so that the part $w_i$ of $\rho$ can
embed. To make sure that there are enough tokens also in places $p$
with $\vec m_i^{\fin}[p]=\omega$, we drive up $k$. Because of
$\sum_{j=0}^i \eff(x_iy_iz_i)[p]\ge 1$ for such places, this creates
enough tokens to embed the run $\rho$.

  In the proof, we will use the concept of a hurdle for a transition
  sequence. Observe that for every sequence $w\in
  T^*$, there exists a smallest marking $\vec m\in\N^P$ such that
  $w$ can be fired in $\vec
  m$.  This is called the \emph{hurdle} of
  $w$ and we denote it by $\hurdle(w)$.
  
  Let $\Ncal=C_0,t_1,C_1,\ldots,t_n,C_n$ and let $\vec m_i$ be the
  distinguished vertex of $C_i$ and let $\vec m_i^{\init}$ be the
  initial vertex of $C_i$.
  Each of the words $u_1,\ldots,u_m$ labels a loop in some precovering
  graph $C_i$. Since $C_i$ is strongly connected, we can pick a loop
  $u'_i\in L(C_i)$ for each $i\in[0,n]$ so that every word $u_j$
  appears as an infix in $u'_1,\ldots,u'_n$.

  Let $\rho$ be the run $w_0t_1w_1\cdots t_nw_n$ in $\Ncal$.  Since
  $\Ncal$ is perfect, there is a covering sequence $x_i\in T^*$ for
  each $C_i$, $i\in[0,n]$. We will now construct a covering sequence
  $x_i^\ell u'_iw_i$ for some $\ell$ and then apply
  \cref{lambert-pumping}.

  Let $\vec n_i$ be the marking that $\rho$ enters after firing $t_i$
  (and hence before firing $w_i$).  Pick $\ell\ge 1$ so that
  \begin{enumerate}
  \item $\ell\ge \hurdle(u'_iw_i)[p]$,
  \item $\ell>|\eff(u'_iw_i)[p]|$, and
  \item $\ell-|\eff(u'_i)[p]|\ge \vec n_i[p]$
  \end{enumerate}
  for every $p\in P$ with $\vec m_i[p]=\omega$. We claim that the
  sequence $x'_i=x_i^\ell u'_iw_i$ is a covering sequence for $C_i$.
  
  First, $x'_i$ belongs to $L(C_i)$, because each word $x_i$, $u'_i$,
  $w_i$ does.  Moreover, $x'_i$ is enabled in $\vec m_i^{\init}$:
  Since $x_i$ is a covering sequence, $x_i^\ell$ is enabled.  Consider
  $p\in P$ with $\vec m_i^{\init}[p]\in\N$. If $\vec m_i[p]=\omega$,
  then $x_i$ has a positive effect on $p$, so that firing $x_i^\ell$
  leaves at least $\ell\ge \hurdle(u'_iw_i)[p]$ tokens in $p$. If
  $\vec m_i[p]\in\N$, then $x_i$, $u'_i$, and $w_i$ have zero effect
  on $p$.  Thus, $x_i^\ell u'_iw_i$ is indeed enabled in
  $\vec m_i^{\init}$.

  To establish that $x'_i$ is a covering sequence, it remains to show
  that for each $p\in P$ one of the conditions (i)--(iii) holds.  This
  is trivial if $\vec m_i^{\init}[p]=\omega$, so suppose
  $\vec m_i^{\init}[p]\in\N$.  If $\vec m_i[p]\in\N$, then each of the
  sequences $x_i$, $u'_i$, and $w_i$ have zero effect on $p$ by virtue
  of belonging to $L(C_i)$. If $\vec m_i[p]=\omega$, then $x_i^\ell$
  produces at least $\ell$ tokens in $p$ and since
  $\ell>|\eff(u'_iw_i)[p]|$, we have $\eff(x_i^\ell u'_i w_i)[p]>0$.
  This completes the proof that $x'_i$ is a covering sequence for $C_i$.

  According to \cref{lambert-pumping}, there are $k_0\in\N$, $\beta_i\in T^*$,
  $y_i,z_i\in T^*$ for $i\in[0,n]$, so that for every $k\ge k_0$, the sequence
  \begin{equation}
    \rho_k=x'^k_0 \beta_0 y_0^k z_0^k\cdot t_1\cdot x'^k_1 \beta_1 y_1^k z_1^k\cdots t_n\cdot x'^k_n\beta_ny_n^k z_n^k \label{run-from-pumping}
  \end{equation}
  is a run in $\Ncal$ for which $\sum_{j=0}^i \eff(x'_jy_jz_j)[p]\ge 1$
  for each $p\in P$ and $i\in[0,n]$ with $\vec
  m_i^{\fin}[p]=\omega$. We now claim that for large enough $k$, the
  run $\rho_k$ can be chosen as the desired $\rho'$.

  The embedding $\sigma$ will map each $t_i$ to the $t_i$ displayed in
  \cref{run-from-pumping}.  Let $\vec m_i^{(k)}$ be the marking
  entered in $\rho_k$ before firing $t_i$. Moreover, $\sigma$ will
  embed each $w_i$ to the infix $w_i$ in the first occurrence of
  $x'_i=x_i^\ell u'_i w_i$. Let $\vec n'^{(k)}_i$ be the marking
  entered before firing this infix $w_i$ in $\rho_k$. To verify that
  this will indeed yield an embedding of runs, we have to show that
  for large enough $k$, we have
  \begin{align}
    \vec n_i-\eff(t_i)\le\vec m_i^{(k)} ~\text{for $i\in[1,n]$ and}~~ \vec n_i\le \vec n'^{(k)}_i~\text{for $i\in[0,n]$} \label{embedding-inequalities}
  \end{align}
  The left inequality states that there are enough tokens to embed the
  marking entered in $\rho$ before firing $t_i$, $i\in[1,n]$. The right inequality
  states that the same is true for the marking entered in $\rho$
  before firing $w_i$, $i\in[0,n]$.

  Let us argue that we can choose $k$ large enough to satisfy
  \cref{embedding-inequalities}. First, consider the left inequality
  for $i\in[1,n]$ and let $p\in P$. If $\vec m_{i-1}^{\fin}[p]\in\N$,
  then of course we have
  $\vec m_i^{(k)}[p]\ge \vec n_i[p]-\eff(t_i)[p]$: We even have
  $\vec m_i^{(k)}[p]=\vec m_{i-1}^{\fin}[p]=\vec n_i[p]-\eff(t_i)[p]$,
  which has to hold for any run in $\Ncal$.  If
  $\vec m_{i-1}^{\fin}[p]=\omega$, then we observe
  \begin{align*}
    \vec m_i^{(k)}[p] &= \eff(x'^k_0\beta_0y_0^kz_0^k)+\sum_{j=1}^{i-1} \eff(t_jx'^k_j\beta_jy_j^kz_j^k) \\
    &= \eff(\beta_0t_1\beta_1\cdots t_{i-1}\beta_{i-1})[p]+k\cdot \sum_{j=0}^{i-1} \eff(x'_jy_jz_j)[p] \\
    &\ge \eff(\beta_0t_1\cdots t_{i-1}\beta_{i-1})[p] + k
  \end{align*}
  because $\sum_{j=0}^{i-1} \eff(x'_jy_jw_j)\ge 1$.  Hence, for large
  enough $k$, the left inequality of \cref{embedding-inequalities} is
  fulfilled for such $i\in[1,n]$.  Thus, it holds for every
  $i\in[1,n]$.

  Now consider the right inequality of \cref{embedding-inequalities}
  for $i\in[0,n]$. Let $p\in P$. We distinguish three cases.
  \begin{enumerate}
  \item Suppose $\vec m_i^{\init}[p]=\omega$. Then we also have
    $\vec m_{i-1}^{\fin}[p]=\omega$. Observe that
    \begin{align*}
      \vec n'_i[p]&=\eff(\beta_0t_1\beta_1\cdots t_{i-1}\beta_{i-1})[p] \\
      &~~~~+ k\cdot\sum_{j=0}^{i-1} \eff(x'_iy_iz_i)[p]+\eff(t_ix_i^\ell u'_i)[p] \\
      &\ge \eff(\beta_0t_1\beta_1\cdots t_{i-1}\beta_{i-1})[p] + k + \eff(t_ix_i^\ell u'_i)[p]
    \end{align*}
    since in this case $\sum_{j=0}^{i-1}\eff(x'_iy_iz_i)\ge 1$. Thus,
    for large enough $k$, we have $\vec n'_i[p]\ge\vec n_i[p]$.
  \item Suppose $\vec m_i^{\init}[p]=\vec m_i[p]\in\N$. Then,
    both $\vec n_i[p]$ and $\vec n'_i[p]$ are necessarily equal to
    $\vec m^{\init}_i[p]$: Such places $p$ are left unchanged by
    cycles in $C_i$. Hence, $\vec n_i[p]=\vec n'_i[p]$.
  \item Suppose $\vec m_i^{\init}[p]\in\N$ and $\vec m_i[p]=\omega$.
    In this case, we have $\vec n_i[p]\le\vec n'_i[p]$ already by our
    choice of $\ell$: Since for such $p$, $x_i$ has a positive effect,
    we have in particular $\vec n'_i[p]\ge \ell$. Since we chose
    $\ell$ with $\ell-|\eff(u'_i)[p]|\ge \vec n_i[p]$, this implies
    that after executing $x_i^\ell u'_i$, there must be at least
    $\vec n_i[p]$ tokens in $\vec n'_i[p]$.
  \end{enumerate}
  Thus, each of the $2n+1$ inequalities in
  \cref{embedding-inequalities} holds for sufficiently large $k$,
  meaning we can choose a $k$ for which they all hold
  simultaneously. With this, let $\rho'=\rho_k$.  Now indeed, $\rho$
  embeds into $\rho'$ via the embedding $\sigma$.  Then clearly
  $u'_i\in F(\sigma)$ for $i\in[0,n]$, hence $u_j\in F(\sigma)$ for
  $j\in[1,m]$.
  \qed
\end{proof}


\paragraph{Proof of the modular envelope property}
We are finally ready to show that $\cA$ is a modular envelope for
$L(\Ncal)$.  Given a run $\rho$ in $\Ncal$ for the transition sequence
$w\in L(\Ncal)$, we first use \cref{inject-factors} to obtain a run
$\rho_1$ such that $\rho_1$ contains each $u_i$, $i\in[1,n]$, as an infix
and $\rho$ embeds via a compatible embedding $\sigma_1$
into $\rho_1$.

For $j\ge 2$, let $\rho_j$ be the run obtained by amalgamating
$\rho_{j-1}$ and $\rho_1$ along $\sigma_{j-1}$ and
$\sigma_1$. Moreover, let $\sigma_j$ be the resulting embedding. Then
\cref{amalgam-in-mgts} tells us that $\rho_j$ is a run in $\Ncal$ for
every $j\ge 2$. Moreover, we have
$u_i\in F(\sigma_{j-1})\subseteq F(\sigma_j)$ for every $i\in[1,m]$
 and also
$\Parikh(\rho_j)-\Parikh(\rho)=\Parikh(\rho_{j-1})-\Parikh(\rho)+\Parikh(\rho_1)-\Parikh(\rho)$
(\cref{amalgamation-properties}) and hence by induction
$\Parikh(\rho_j)=\Parikh(\rho)+j\cdot
(\Parikh(\rho_1)-\Parikh(\rho))$.  In particular
$\Parikh(\rho_k)\equiv\Parikh(\rho)\bmod{k}$. Therefore, the
transition sequence of $\rho_k$ is a word $w'$ that has each $u_i$ as
an infix and satisfies $\Parikh(w')\equiv\Parikh(w)\bmod{k}$. This
proves that $\cA$ is a modular envelope for $L(\Ncal)$.

\bibliography{citat}

\end{document}